\documentclass[11pt, letterpaper]{article}
\usepackage[utf8]{inputenc}
\usepackage{fullpage,times}

\usepackage{fullpage}

\usepackage[colorlinks=false,citecolor=black]{hyperref}
\usepackage{amsmath,amsfonts,amsthm,amssymb,multirow}
\usepackage{times}
\usepackage{graphicx}
\usepackage{floatpag}
\usepackage{algorithm}
\usepackage[noend]{algpseudocode}
\usepackage{bbold}
\usepackage{enumitem}
\usepackage{subcaption} 
\usepackage{calc}
\usepackage{tikz}
\usetikzlibrary{decorations.markings,decorations.pathreplacing,arrows, automata, positioning, quotes, shapes,backgrounds, arrows.meta}
\tikzstyle{vertex}=[circle, draw, inner sep=0pt, minimum size=4pt, fill = black]
\usepackage[numbers,sort]{natbib}

\usepackage{graphicx}
\usepackage{tabularx}
\makeatletter
\newcommand{\multiline}[1]{%
  \begin{tabularx}{\dimexpr\linewidth-\ALG@thistlm}[t]{@{}X@{}}
    #1
  \end{tabularx}
}
\makeatother
\usepackage{mathtools}

\makeatletter
\def\BState{\State\hskip-\ALG@thistlm}
\makeatother

\newcommand{\ceil}[1]{\lceil #1 \rceil}
\newcommand{\floor}[1]{\lfloor #1 \rfloor}

\usepackage[compact]{titlesec}
\titlespacing{\section}{0pt}{3ex}{2ex}
\titlespacing{\subsection}{0pt}{2ex}{1ex}
\titlespacing{\subsubsection}{0pt}{0.5ex}{0ex}

\newtheorem{theorem}{Theorem}[section]

\newtheorem{corollary}[theorem]{Corollary}

\newtheorem{lemma}[theorem]{Lemma}
\newtheorem{claim}[theorem]{Claim}

\newtheorem{problem}[theorem]{Problem}

\makeatletter
\let\c@fconjecture\c@conjecture
\makeatother

\makeatletter
\let\c@fconj\c@conj
\makeatother

\def \qed {\hfill $\Box$}

\def \eps {\varepsilon}

\newcommand{\ignore}[1]{}

\def \polylog { \text{\rm polylog~} }

\def\tO{\tilde{O}}

\newcommand{\D}{\Delta}
\newcommand{\OO}{\tilde{O}}
\newcommand{\MMM}{{\cal M}}
\newcommand{\MM}{{\cal M}^\star}

\newcommand{\MMp}{{\mbox{\rm M}^\star}}

\newcommand{\IGNORE}[1]{}

\newcommand{\pmM}{$[\pm \ccc]$}
\newcommand{\ppM}{$[\ccc]$}
\newcommand{\pnzM}{$([\ccc]-\{0\})$}
\newcommand{\ccc}{c_0}

\newcommand{\DD}[1]{D_2^{(#1)}}

\newcommand{\CountCap}[1]{$\#_{\le #1}$APSP}
\newcommand{\CountMod}[1]{$\#_{\mbox{\scriptsize\rm mod}\, #1}$APSP}
\newcommand{\CountApx}[1]{$\#_{\mbox{\scriptsize\rm approx-}#1}$APSP}
\newcommand{\RedAPSP}[1]{u-$#1$Red-APSP}
\newcommand{\wmax}{M}

\DeclareMathOperator*{\argmin}{arg\,min}

\title{Algorithms, Reductions and Equivalences for Small Weight Variants of All-Pairs Shortest Paths}

\author{Timothy M. Chan\\UIUC\\tmc@illinois.edu \and Virginia {Vassilevska Williams}\\MIT\\virgi@mit.edu \and Yinzhan Xu\\MIT\\xyzhan@mit.edu}

\begin{document}
\date{}

\maketitle

\begin{abstract}
All-Pairs Shortest Paths (APSP) is one of the most well studied problems in graph algorithms. This paper studies several variants of APSP in unweighted graphs or graphs with small integer weights.

APSP with small integer weights in undirected graphs [Seidel'95, Galil and Margalit'97] has an $\tO(n^\omega)$ time algorithm, where $\omega<2.373$ is the matrix multiplication exponent. APSP in directed graphs with small weights however, has a much slower running time that would be $\Omega(n^{2.5})$ even if $\omega=2$ [Zwick'02]. To understand this $n^{2.5}$ bottleneck, we build a web of reductions around directed unweighted APSP\@. We show that it is {\em fine-grained equivalent} to computing a rectangular Min-Plus product for matrices with integer entries; the dimensions and entry size of the matrices depend on the value of $\omega$. As a consequence, we establish an equivalence between APSP in directed unweighted graphs, APSP in directed graphs with small $(\OO(1))$ integer weights, All-Pairs Longest Paths in DAGs with small weights, $c$Red-APSP in undirected graphs with small weights, for any $c\geq 2$
(computing all-pairs shortest path distances among paths that use at most $c$ red edges), $\#_{\le c}$APSP in directed graphs with small weights (counting the number of shortest paths for each vertex pair, up to $c$), and approximate APSP with additive error $c$ in directed graphs with small weights, for $c\le \OO(1)$. 

We also provide fine-grained reductions from directed unweighted APSP to All-Pairs Shortest Lightest Paths (APSLP) in undirected graphs with $\{0,1\}$ weights and $\#_{\mbox{\scriptsize mod}\, c}$APSP in directed unweighted graphs (computing counts mod $c$), thus showing that unless the current algorithms for APSP in directed unweighted graphs can be improved substantially, these problems need at least $\Omega(n^{2.528})$ time.

We complement our hardness results with new algorithms. We improve the known algorithms for APSLP in directed graphs with small integer weights (previously studied by Zwick [STOC'99]) and for approximate APSP with sublinear additive error in directed unweighted graphs (previously studied by Roditty and Shapira [ICALP'08]).
Our algorithm for approximate APSP with sublinear additive error is optimal, when viewed as a reduction to Min-Plus product.  We also give new algorithms for variants of \#APSP (such as $\#_{\le U}$APSP and $\#_{\mbox{\scriptsize mod}\, U}$APSP for $U\le n^{\OO(1)}$) in unweighted graphs, as well as a near-optimal $\OO(n^3)$-time algorithm for the original \#APSP problem in unweighted graphs (when counts may be exponentially large). This also implies an $\OO(n^3)$-time algorithm for Betweenness Centrality, improving on the previous $\OO(n^4)$ running time for the problem. 
Our techniques also lead to a simpler alternative to Shoshan and Zwick's algorithm [FOCS'99] for the original APSP problem in undirected graphs with small integer weights.
\end{abstract}

\section{Introduction}
\emph{All-Pairs Shortest Paths (APSP)} is one of the oldest and most studied problems in graph algorithms. The fastest known algorithm for general $n$-node  graphs runs in $n^3/2^{\Theta(\sqrt{\log n})}$ \cite{ryanapsp}. In unweighted graphs, or graphs with small integer weights, faster algorithms are known. 

For APSP in {\em undirected} unweighted graphs (u-APSP), Seidel \cite{Seidel} and
Galil and Margalit~\cite{GalilMargalitAPSD,GalilMargalitAPSP}
gave  $\tilde{O}(n^\omega)$ time algorithms where $\omega\leq 2.373$ is the exponent of matrix multiplication \cite{almanvw21,vstoc12,legallmult}; the latter algorithm 
works for graphs with small integer weights\footnote{
In this paper, $[\pm c_0]=\{-c_0,\ldots,c_0\}$ and $[c_0]=\{0,\ldots,c_0\}$.
The $\OO$ notation hides polylogarithmic factors (although
conditions of the form $\ccc=\OO(1)$
may be relaxed to $\ccc\le n^{o(1)}$ if we allow extra $n^{o(1)}$ factors in the $\OO$ time bounds).
} 
in \pmM\ for
$\ccc=\OO(1)$.
The hidden dependence on $\ccc$
was improved by
Shoshan and Zwick \cite{shoshanzwick}.

For {\em directed} unweighted graphs or graphs with weights in \pmM, the fastest APSP algorithm is by Zwick \cite{zwickbridge}, running in $O(n^{2.529})$ time. This running time is achieved using the best known bounds for rectangular matrix multiplication \cite{legallurr} and would be $\Omega(n^{2.5})$ even if $\omega=2$.

There is a big discrepancy between the running times for undirected and directed APSP\@. One might wonder, why is this? Are directed graphs inherently more difficult for APSP, or is there some special graph structure we can uncover and then use it to develop an $\tilde{O}(n^\omega)$ time algorithm for directed APSP as well? (Note that matrix multiplication seems necessary for APSP since APSP is known to capture Boolean matrix multiplication.)

The first contribution in this paper is a fine-grained equivalence between directed unweighted APSP (u-APSP) and a certain rectangular version of the Min-Plus product problem. 

The \emph{Min-Plus product} of an $n\times m$ matrix $A$ by an $m\times p$ matrix $B$ is the matrix $C$ with entries $C[i,j]=\min_{k=1}^m (A[i,k]+B[k,j])$. 
Let us denote by $\MMp(n_1,n_2,n_3 \mid M)$ the problem of computing the Min-Plus product of an $n_1\times n_2$ matrix by an $n_2\times n_3$ matrix where both matrices have integer entries in $[M]$. Let $\MM(n_1,n_2,n_3 \mid M)$ be the best running time for $\MMp(n_1,n_2,n_3 \mid M)$.

Zwick's algorithm~\cite{zwickbridge} for u-APSP can be viewed as making a logarithmic number of calls to the Min-Plus product 
$\MMp(n,n/L,n \mid L)$ for all $1\leq L\leq n$ that are powers of $3/2$. The running time of Zwick's algorithm is thus, within polylogarithmic factors, $\max_{L} \MM(n,n/L,n \mid L)$.

Let $\MMM(a,b,c)$ denote the running time of the fastest algorithm to multiply an $a\times b$ by a $b\times c$ matrix over the integers. Let $\omega(a,b,c)$ be the smallest real number $r$ such that $\MMM(n^a,n^b,n^c)\leq O(n^{r+\eps})$ for all $\eps>0$.


The best known upper bound for the Min-Plus product running time $\MM(n,n/L,n \mid L)$ is the minimum of $O(n^3/L)$ (the brute force algorithm) and $\tilde{O}(L\cdot \MMM(n,n/L,n))$ \cite{AlonGalilMargalit}. For $L=n^{1-\ell}$,  $\MM(n,n/L,n \mid L)$ is thus at most $\tilde{O}(\min\{n^{2+\ell},n^{1-\ell+\omega(1,\ell,1)}\})$. 
Over all $\ell\in [0,1]$, the runtime is maximized at $\tilde{O}(n^{2+\rho})$ where $\rho$ is such that $\omega(1,\rho,1)=1+2\rho$.


Hence in particular, the running time of Zwick's algorithm is $\tilde{O}(n^{2+\rho})$. 
This running time has remained unchanged (except for improvements on the bounds on $\rho$) for almost 20 years. The current best known bound on $\rho$ is $\rho<0.529$, and if $\omega=2$, then $\rho=1/2$.

Our first result is that u-APSP is sub-$n^{2+\rho}$ fine-grained equivalent to $\MMp(n,n^\rho,n \mid n^{1-\rho})$:
\begin{theorem}\label{thm:equiv1}
If $\MMp(n,n^\rho,n \mid n^{1-\rho})$ is in $O(n^{2+\rho-\eps})$ time for $\eps>0$, then u-APSP can also be solved in $O(n^{2+\rho-\eps'})$ time for some $\eps'>0$.
If u-APSP can can be solved in $O(n^{2+\rho-\eps})$ time for some $\eps>0$, then $\MMp(n,n^\rho,n \mid n^{1-\rho})$ can also be solved in $O(n^{2+\rho-\eps})$ time.
\end{theorem}

The Min-Plus product of two $n\times n$ matrices with {\em arbitrary} integer entries is known to be equivalent to APSP with {\em arbitrary} integer entries \cite{fischermeyer}, so that their running times are the same, up to constant factors. 
All known algorithms for directed unweighted APSP (including \cite{zwickbridge,AlonGalilMargalit} and others), make calls to Min-Plus product of rectangular matrices with integer entries that can be as large as say $n^{0.4}$.
It is completely unclear, however, why a problem in {\em unweighted} graphs such as u-APSP should require the computation of Min-Plus products of matrices with such large entries. Theorem~\ref{thm:equiv1} surprisingly shows that it does. Moreover, it shows that unless we can improve upon the known approaches for Min-Plus product computation, there will be no way to improve upon Zwick's algorithm for u-APSP\@. The latter is an algebraic problem in disguise.

The main proof of Theorem~\ref{thm:equiv1} is  simple---what is remarkable are the numerous consequences on equivalences and  conditional hardness that follow from this idea.
We first use Theorem~\ref{thm:equiv1} to build a class of problems that are all equivalent to u-APSP, via $(n^{2+\rho},n^{2+\rho})$-fine-grained reductions (see \cite{vsurvey} for a survey of fine-grained complexity). In particular, if $\omega=2$ (or more generally when $\omega(1, \frac{1}{2}, 1) = 2$), these are all problems that are $n^{2.5}$-fine-grained equivalent.

\IGNORE{
Here is some APSP-related problems we consider:
\begin{itemize}
    \item All-Pairs Longest Paths (APLP): output for every pair of vertices $s,t$ the weight of the longest path from $s$ to $t$ in $G$;
    \item \#APSP: output for every pair of vertices $s,t$ the number $C[s,t]$ of shortest paths from $s$ to $t$;
    \item \#$_{\le c}$APSP: similar to \#APSP, except that we return $\max\{C[s,t],c\}$  (think of $c$ as a cap); note that for $c=2$, this corresponds to testing uniqueness of the shortest path for each pair;
    \item \RedAPSP{c}: for a graph in which some edges are colored red and others are colored blue,  output for every pair of vertices $s,t$ the weight of the shortest path from $s$ to $t$ that uses at most $c$ red edges.
\end{itemize}
}

Recall that in the
\emph{All-Pairs Longest Paths (APLP)} problem, we want to output for every pair of vertices $s,t$ the weight of the longest path from $s$ to $t$. While APLP is NP-hard in general, it is efficiently solvable in DAGs.
In the
\emph{$c$Red-APSP} 
problem, for a given graph in which some edges can be colored red, we want to output for every pair of vertices $s,t$ the weight of the shortest path from $s$ to $t$ that uses at most $c$ red edges. For convenience, we call all non-red edges blue. 

We use the following convention for problem names:
the prefix ``u-'' is for unweighted graphs;
the prefix ``\ppM-'' is for graphs with weights in \ppM\ (similarly for ``\pmM-'' and for other ranges).
Input graphs are directed unless stated otherwise.

\begin{theorem}
\label{thm:intro:equiv2}
The following problems either all have $O(n^{2+\rho-\eps})$ time algorithms for some $\eps>0$, or none of them do, assuming
that $c_0=\OO(1)$:
\begin{itemize}
\setlength\itemsep{0em}
\item $\MMp(n,n^\rho,n \mid n^{1-\rho})$,
\item u-APSP,
\item \pmM-APSP for directed graphs without negative cycles,
\item u-APLP for DAGs,
\item \pmM-APLP for DAGs,
\item \RedAPSP{c} for \emph{undirected} graphs for any  $2\le c\le \OO(1)$.
\end{itemize}
\end{theorem}

\IGNORE{

\begin{theorem}
\label{thm:intro:equiv2}
The following problems either all have $O(n^{2+\rho-\eps})$ time algorithms for some $\eps>0$, or none of them do:
\begin{itemize}
\setlength\itemsep{0em}
\item $\MMp(n,n^\rho,n \mid n^{1-\rho})$,
\item u-APSP,
\item M-APSP: APSP in directed graphs with weights in $\{-M,\ldots,M\}$ for $M\leq \tilde{O}(1)$ and no negative cycles,
\item All-Pairs Longest Paths in unweighted DAGs (DAG-u-APLP): given an unweighted directed acyclic graph (DAG) $G$, output for every pair of nodes $s,t$ the weight of the longest path from $s$ to $t$ in $G$,
\item  
DAG-$M$-APLP: All-Pairs Longest Paths in DAGs with with weights in $\{-M,\ldots,M\}$ for $M\leq \tilde{O}(1)$,
\item \CountCap{U}: given an unweighted graph and an $n^{o(1)}$-bit integer $U$, return for every pair of vertices $u,v$, the minimum of $U$ and the number of shortest paths between $u$ and $v$,
\item \RedAPSP{c}: given an unweighted \textbf{undirected} graph in which some edges can be colored red and any $n^{o(1)}$-bit integer $c \ge 2$, return for every pair of vertices, the weight of the shortest path between them that uses at most $c$ red edges. For convenience, we will call all the not red edges blue edges. 
\end{itemize}
\end{theorem}

}

Interestingly, while \RedAPSP{2} in undirected graphs above is equivalent to u-APSP and hence improving upon its $\tilde{O}(n^{2+\rho})$ runtime would be difficult, we show that \RedAPSP{1} in undirected graphs can be solved in $\tilde{O}(n^\omega)$ time via a modification of Seidel's algorithm, and hence there is a seeming jump in complexity in \RedAPSP{c} from $c=1$ to $c=2$.

Besides the above equivalences we provide some interesting reductions from u-APSP to other well-studied matrix product and shortest paths problems.

Lincoln, Polak and Vassilevska W.~\cite{lincoln2020monochromatic} reduce u-APSP to some matrix product problems such as All-Edges Monochromatic Triangle and the $(\min, \max)$-Product studied in \cite{VassilevskaWY06,VassilevskaWY10} and \cite{VassilevskaWY07,duanpettiebott} respectively. Using the equivalence of u-APSP and $\MMp(n,n/\ell,n\mid \ell^{1-p})$, we can reduce u-APSP to another matrix product called Min Witness Equality Product (MinWitnessEq), where we are given $n \times n$ integer matrices $A$ and $B$, and are required to compute $\min\{k \in [n]: A[i,k]=B[k, j]\}$ for every pair of $(i, j)$. This can be viewed as a merge of the Min Witness product \cite{CzumajKL07} \footnote{Recently, there has been renewed interest in studying the Min Witness product, due to a breakthrough \cite{aplca} on the All-Pairs LCA in DAGs problem, which was one of the original motivations for studying Min Witness.} and Equality Product problems \cite{labib2019hamming,vnotes}.

Another natural variant of APSP is the problem of \emph{approximating} shortest path distances. Zwick~\cite{zwickbridge} presented an $\tilde{O}(n^\omega \log M)$ time
algorithm for computing a $(1+\eps)$-multiplicative approximation for all pairwise distances in a directed graph with integer weights in $[M]$, for any constant $\eps>0$.%
\footnote{
Bringmann et al.~\cite{bringmannmaxmin} considered the more unusual setting of very large $\wmax$, where the $\log \wmax$ factor is to be avoided.  
}
This is essentially optimal since any such approximation algorithm can be used to multiply $n\times n$ Boolean matrices. 

An arguably better notion of approximation is to provide an additive approximation, i.e. outputting for every $u,v$ an estimate $D'[u,v]$ for the distance $D[u,v]$ such that $D[u,v]\leq D'[u,v]\leq D[u,v]+E$, where $E$ is an error that can depend on $u$ and $v$. 


At ICALP'08,
Roditty and Shapira \cite{RodittyShapira} studied the following variant: given an unweighted directed graph and a constant $p\in [0,1]$,  compute for all $u,v$ an estimate $D'[u,v]$ with $D[u,v]\leq D'[u,v]\leq D[u,v]+D[u,v]^p$.
They gave an algorithm with running time
$\OO(\max_\ell \min\{n^3/\ell,\MM(n,n/\ell^{1-p},n\mid \ell^{1-p})\})$.
For example, for $p=0$, this matches the time complexity of Zwick's exact
algorithm for u-APSP; for $p=1$, this matches Zwick's $\OO(n^\omega)$-time algorithm with constant multiplicative approximation factor.
For $p=0.5$, with the current rectangular matrix multiplication bounds~\cite{legallurr}, the running time is $O(n^{2.447})$.

We obtain an improved running time:
\begin{theorem}
\label{thm:intro:additive}
For any $p\in [0,1]$, given a directed unweighted graph, one can obtain additive $D[u,v]^p$ approximations to all distances $D[u,v]$ in time $\OO(\max_\ell \MM(n,n/\ell,n\mid \ell^{1-p}))$.
\end{theorem}

The improvement over Roditty and Shapira's running time is substantial.
For example, for all $p\ge 0.415$, the time bound is $O(n^{2.373})$ (the current matrix multiplication running time), whereas their algorithm only achieves $O(n^{2.373})$ for $p=1$.
Our result also answers one of Roditty and Shapira's open question (on whether $\OO(n^\omega)$ time is possible for any $p<1$), if $\omega>2$.

The new algorithm is also \emph{optimal} (ignoring logarithmic factors) in a strong sense, as our reduction technique shows that for all $\ell$, $\MMp(n,n/\ell,n\mid \ell^{1-p})$ can be tightly reduced to the additive $D[u,v]^p$ approximation of APSP\@.  
In particular, 
u-APSP with constant additive error is fine-grained equivalent to exact u-APSP\@.

The \emph{All-Pairs Lightest Shortest Paths (APLSP)} problem studied in \cite{ChanSTOC07,ZwickSTOC99} asks 
to compute for every pair of vertices $s,t$  
the distance from $s$ to $t$ (with respect to the edge weights) and the smallest number of edges over all shortest paths from $s$ and $t$. 
Traditional shortest-path algorithms can be easily modified to find the lightest shortest paths, but not the faster matrix-multiplication-based algorithms.
Our reduction for \RedAPSP{c} can be easily modified to reduce 
$\MMp(n,n^\rho,n \mid n^{1-\rho})$ to $\{0,1\}$-APLSP in undirected graphs, which can be viewed as a conditional lower bound of $n^{2+\rho-o(1)}$ for the latter problem.

\begin{corollary}\label{thm:aplsplb}
If $\{0,1\}$-APLSP in undirected graphs is in $O(n^{2+\rho-\eps})$ time for $\eps>0$, then so is $\MMp(n,n^\rho,n \mid n^{1-\rho})$.
\end{corollary}

The fastest known algorithm to date for $\{0,1\}$-APLSP,
or more generally, \ppM-APLSP for $\ccc=\OO(1)$, for directed or undirected graphs
is by Zwick~\cite{ZwickSTOC99} from STOC'99 and runs in $O(n^{2.724})$ time with the current best bounds for rectangular matrix multiplication (the running time would be $\OO(n^{8/3})$ if $\omega=2$). Chan~\cite{ChanSTOC07} (STOC'07) improved this running time to $\tilde{O}(n^{(3+\omega)/2})\leq O(n^{2.687})$ 
but only if the weights are {\em positive}, i.e., for 
\pnzM-APLSP
(and so his result does not hold for $\{0,1\}$-APLSP).

Both Zwick's and Chan's algorithms solve a more general problem, Lex$_2$-APSP, in which one is given a directed graph where each edge $e$ is given two weights $w_1(e),w_2(e)$ and one wants to find for every pair of vertices $u,v$ the lexicographic minimum over all $u$-$v$ paths $\pi$ of $(\sum_{e\in\pi}w_1(e),\sum_{e\in \pi}w_2(e))$. Then APLSP is Lex$_2$-APSP when all $w_2$ weights are $1$, and the related \emph{All-Pairs Shortest Lightest Paths (APSLP)} problem is when all $w_1$ weights are $1$.

To complement the conditional lower bound for APLSP, and hence Lex$_2$-APSP, we present new algorithms for \ppM-Lex$_2$-APSP for $\ccc=\OO(1)$, both (slightly) improving Chan's running time and also allowing zero weights, something that Chan's algorithm couldn't support.

\begin{theorem}
\label{thm:intro:lexapsp}
\ppM-Lex$_2$-APSP can be solved in $O(n^{2.66})$ time for any $\ccc=\OO(1)$.
\end{theorem}

If $\omega=2$, the above running time would be $\tilde{O}(n^{2.5})$,
improving Zwick's previous $\tilde{O}(n^{8/3})$ bound~\cite{ZwickSTOC99} and matching our conditional lower bound $n^{2+\rho-o(1)}$.  For undirected graphs with {\em positive} weights in $[\ccc]-\{0\}$,
we further improve the running time to $O(n^{2.58})$ under the current matrix multiplication bounds.

We next consider the natural problem, \emph{$\#$APSP}, of counting the number of shortest paths for every pair of vertices in a graph. This problem needs to be solved, for example, when computing the so-called \emph{Betweenness Centrality (BC)} of a vertex. BC is a well-studied measure of vertex importance in social networks. If we let $C[s,t]$ be the number of shortest paths between $s$ and $t$, and $C_v[s,t]$ be the number of shortest paths between $s$ and $t$ that go through $v$, then $\text{BC}(v)=\sum_{s,t\neq v} C_v[s,t]/C[s,t]$ and the BC problem is to compute $\text{BC}(v)$
for a given graph and a given node $v$. 

Prior work \cite{Brandes} showed that $\#$APSP and BC in $m$-edge $n$-node unweighted graphs can be computed in $O(mn)$ time via a modification of Breadth-First Search (BFS).\footnote{Brandes presented further practical improvements as well.} However, all prior algorithms assumed a model of computation where adding two integers of arbitrary size takes constant time. In the more realistic word-RAM model (with $O(\log n)$ bit words), these algorithms would run in $\tilde{\Theta}(mn^2)$ time, as there are explicit examples of graphs with $m$ edges (for any $m$, a function of $n$) for which the shortest paths counts have $\Theta(n)$ bits.\footnote{One example is an $(n/3+2)$-layered graph where the first $n/3$ layers have $2$ vertices each and the last $2$ layers have $n/6$ vertices each. The $i$-th layer and the $(i+1)$-th layer are connected by a complete bipartite graph for each $1 \le i \le n/3$, while the last two layers are connected by $O(m)$ edges. }
In particular, the best running time in terms of $n$ so far has been $\tilde{O}(n^4)$.

We provide the first genuinely $\OO(n^3)$ time algorithm for $\#$APSP, and thus Betweenness Centrality, in directed unweighted graphs.
\begin{theorem}
\label{thm:intro:countapsp}
u-$\#$APSP can be solved in $\OO(n^3)$ time by a combinatorial algorithm.
\end{theorem}
This runtime cannot be improved since there are graphs for which the output size is $\Omega(n^3)$. 

Since the main difficulty of the $\#$APSP problem comes from the counts being very large, it is interesting to consider variants that mitigate this. Let $U\le n^{\OO(1)}$. 
Let \CountMod{U} be the problem of computing all pairwise counts modulo $U$. Let \CountCap{U} be the problem of computing for every pair of nodes $u,v$ the minimum of their count and $U$ (think of $U$ as a ``cap''). Finally, let \CountApx{U} be the problem of computing a $(1+1/U)$-approximation of all pairwise counts (think of keeping the $\log U$ most significant bits of each count).

We obtain the following result for u-\CountCap{U} in directed graphs:

\begin{theorem}\label{thm:apspcapu}
u-\CountCap{U} (in directed graphs) can be solved in $n^{2+\rho}\polylog U\leq n^{2+\rho+o(1)}$ time.

Furthermore, for any $U\ge 2$, if
 u-\CountCap{U} can be solved in $O(n^{2+\rho-\eps})$ time for some $\eps>0$, then so can u-APSP (with randomization).
For any $2\le U\le \OO(1)$, the converse is true as well.
%
\end{theorem}


Thus, we get a conditionally optimal algorithm for u-\CountCap{U}\@.
For $2\le U\le \OO(1)$, the theorem above gives a fine-grained equivalence between u-\CountCap{U} and u-APSP; in particular, for $U=2$, the problem corresponds to testing \emph{uniqueness} for the shortest path of each pair.
(For large $U$, however,
it is not a fine-grained equivalence since the algorithm for u-\CountCap{U} does not go through Min-Plus product, but rather directly uses fast matrix multiplication.) 

Our algorithm from Theorem~\ref{thm:apspcapu} is based on Zwick's algorithm for u-APSP\@. We show that one can also modify Seidel's algorithm for u-APSP in undirected graphs to obtain $\OO(n^\omega)$ time algorithms for u-\CountCap{U} and u-\CountMod{U} in undirected graphs.

\begin{theorem}\label{thm:apspundu}
u-\CountCap{U} and u-\CountMod{U} in undirected graphs can be solved in $\tilde{O}(n^\omega \log U)$ time.
\end{theorem}

We also show that u-\CountApx{U} in undirected graphs can be solved in $O(n^{2.58}\polylog U)$ time, somewhat surprisingly, by a slight modification of our undirected Lex$_2$-APSP algorithm (despite the apparent dissimilarity between the two problems).

\subparagraph*{Paper Organization and Techniques.}
In Section~\ref{sec:reductions}, we show the web of reductions around u-APSP, proving Theorem~\ref{thm:equiv1}, Theorem~\ref{thm:intro:equiv2}, the hardness of additive $D[u,v]^p$ approximate u-APSP and the hardness of u-\CountCap{U} in Theorem~\ref{thm:apspcapu}.  

In Section~\ref{sec:approx}, we give our algorithms for approximating APSP with additive errors, proving Theorem~\ref{thm:intro:additive}. 
In Appendix~\ref{sec:aplsp}, we describe our algorithms for Lex$_2$-APSP\@. In Appendix~\ref{sec:counting}, we consider various versions of $\#$APSP\@. In Appendix~\ref{sec:counting:directedcap}, we prove Theorem~\ref{thm:apspcapu}. In Appendix~\ref{sec:counting:undirected}, we prove Theorem~\ref{thm:apspundu}. In Appendix~\ref{sec:counting:undirectedapprox}, we give an algorithm for u-\CountApx{U}\@. Finally, In Appendix~\ref{sec:counting:exact}, we give an algorithm for u-\CountCap{U} to complete the proof of Theorem~\ref{thm:apspcapu}. 


For approximating APSP with additive error,
we propose an interesting \emph{two-phase}
variant of Zwick's algorithm~\cite{zwickbridge}.  Zwick's algorithm computes distance products of $n\times (n/\ell)$ with $(n/\ell)\times n$ matrices for $\ell$ in a geometric progression.  Our idea is to do less during the first phase, computing products of $(n/\ell)\times (n/\ell)$ with $(n/\ell)\times n$ matrices instead.  We complete the work during a second phase.  The observation is that for the APSP approximation problem, we can afford to perform the distance computation in the first phase \emph{exactly}, but use approximation to speed up the second phase.
The resulting approximation algorithm is even simpler than
Roditty and Shapira's previous (slower) algorithm~\cite{RodittyShapira}.

Our Lex$_2$-APSP algorithm for directed
graphs also uses this two-phase approach, but
in a more sophisticated way to control the size of the numbers in the rectangular matrix products.  A number of interesting new ideas are needed.

To further illustrate the power of this two-phase approach, we also show (in Appendix~\ref{sec:simple:undir}) how the idea can lead to an alternative $\OO(\ccc n^\omega)$ time algorithm for the standard $[\ccc]$-APSP problem for undirected graphs, rederiving Shoshan and Zwick's result~\cite{shoshanzwick} in an arguably simpler way.  This may be of independent interest (as Shoshan and Zwick's algorithm has complicated details).

Our Lex$_2$-APSP algorithm for \emph{undirected} graphs uses small dominating sets for high-degree vertices, an idea of Aingworth et al.~\cite{AingworthCIM99}.  Originally, this idea was for developing combinatorial algorithms for \emph{approximate} shortest paths that avoid matrix multiplication.  Interestingly, we show that this idea can be combined with (rectangular) matrix multiplication to compute \emph{exact} Lex$_2$ shortest paths.

\section{Preliminaries}
The computation model of all algorithms and reductions in this paper is the word-RAM model with $O(\log n)$ bit words.

We let $\MMM(n_1,n_2,n_3)$ denote the best known running time for multiplying an $n_1\times n_2$ by an $n_2\times n_3$ matrix over the integers.
We use $\omega(a, b, c)$ to denote the rectangular matrix multiplication exponent, i.e. the smallest real number $z$ such that $\MMM(n^a,n^b,n^c)\leq O(n^{z+\eps})$ for all $\eps>0$.
In particular, let $\omega=\omega(1,1,1)$. It is known that $\omega\in [2,2.373)$ \cite{almanvw21,vstoc12,legallmult}. The best known bounds for $\omega(a,b,c)$ are in \cite{legallurr}.


Let $\MM(n_1,n_2,n_3\mid \ell_1,\ell_2)$ be the time to compute the Min-Plus product
of an $n_1\times n_2$ matrix $A$ with an $n_2\times n_3$ matrix $B$, where all 
finite entries of $A$ are from $[\ell_1]$ and all finite entries of $B$ are from $[\ell_2]$.  
Let us also denote $\MM(n_1,n_2,n_3\mid \ell) :=
\MM(n_1,n_2,n_3\mid \ell,\ell)$.
It is known~\cite{AlonGalilMargalit}
that  $\MM(n_1,n_2,n_3\mid \ell) \le \OO(\ell\cdot \MMM(n_1,n_2,n_3))$. This algorithm in~\cite{AlonGalilMargalit} first replaces each entry $e$ in both matrices $A, B$ by $(n_2+1)^e$, then uses fast rectangular matrix multiplication to compute the product of the new matrices $A, B$. Since each arithmetic operation takes $\Tilde{O}(\ell)$ time, the running time follows. 

More generally,
let $\MM(n_1,n_2,n_3 \mid m_1,m_2,m_3\mid \ell_1,\ell_2)$ be the time to compute $m_3$ given entries of
the Min-Plus product of an $n_1\times n_2$ matrix $A$ with an $n_2\times n_3$ matrix $B$, where $A$ has at most $m_1$ finite entries, all from $[\ell_1]$,
and $B$ has at most $m_2$ finite entries, all from $[\ell_2]$.



\section{Directed APSP and Rectangular Min-Plus with Bounded Weights}
\label{sec:reductions}
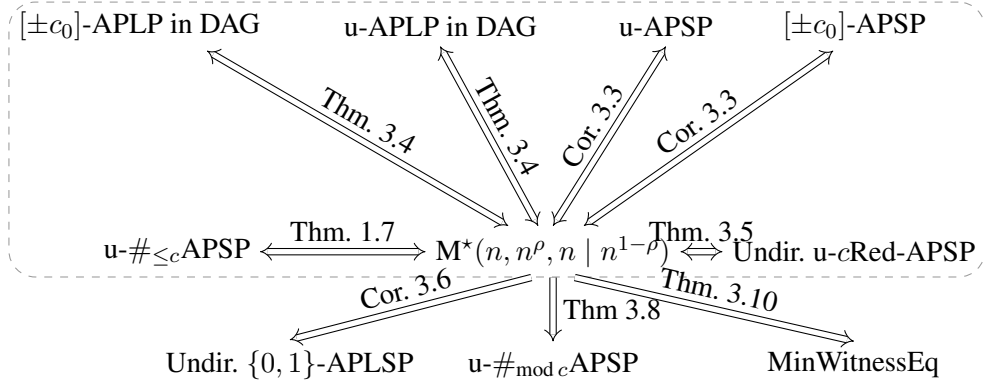
\begin{figure}[ht]
    \centering
    \begin{tikzpicture}
        \node at(1.5, 3)  [anchor=center] (unweightedAPSP){u-APSP};
        
        \node at(4, 3)  [anchor=center] (capsp){\pmM-APSP};

        \node at(-1.5, 3)  [anchor=center] (dagAPLP){u-APLP in DAG};
        
        \node at(-5.5, 3)  [anchor=center] (cdagAPLP){\pmM-APLP in DAG};

        \node at(4, 0)  [anchor=center]
        (2red){ Undir.\ \RedAPSP{c}};
        \node at(0, 0)  [anchor=center] (rectMinPlus){$\MMp(n,n^\rho,n\mid n^{1-\rho})$};
        \node at(-5, 0)  [anchor=center] (capapsp){u-\CountCap{c}};

        \draw[Implies-Implies,double distance=2pt] (unweightedAPSP.south) to[]  node[sloped, anchor=center, above] {Cor.~\ref{cor:unweightedAPSP_equal_MinPlus}} (rectMinPlus.90);
        \draw[Implies-Implies,double distance=2pt] (capsp.230) to[]  node[sloped, anchor=center, above] {Cor.~\ref{cor:unweightedAPSP_equal_MinPlus}} (rectMinPlus.40);
        \draw[Implies-Implies,double distance=2pt] (dagAPLP.south) to[]  node[sloped, anchor=center, above] {Thm.~\ref{thm:APLPdag_equal_Min_Plus}} (rectMinPlus.120);
        \draw[Implies-Implies,double distance=2pt] (cdagAPLP.340) to[]  node[sloped, anchor=center, above] {Thm.~\ref{thm:APLPdag_equal_Min_Plus}} (rectMinPlus.150);
        
        \draw[Implies-Implies,double distance=2pt] (2red.west) to[]  node[sloped, anchor=center, above] {Thm.~\ref{thm:2redAPSP}} (rectMinPlus.east);
        \draw[Implies-Implies,double distance=2pt] (capapsp.east) to[]  node[sloped, anchor=center, above] {Thm.~\ref{thm:apspcapu}} (rectMinPlus.west);

        \draw[opacity=0.4, dashed, rounded corners=10] (current bounding box.north east) -- (current bounding box.north west) -- (current bounding box.south west) -- (current bounding box.south east) -- cycle;
        
        \node[] at(current bounding box.305)  [] (equivclass){\ \ \ \ \ };
        
        \node at(-3.5, -1.5)  [anchor=center] (aplsp){Undir.\ $\{0,1\}$-APLSP};
        
        \node at(0, -1.5)  [anchor=center] (modp){u-\CountMod{c}};
        
        \node at(4, -1.5)  [anchor=center] (minwitnesseq){MinWitnessEq};

        \draw[-Implies,double distance=2pt] (rectMinPlus.230) to[]  node[sloped, anchor=center, above] {Cor.~\ref{cor:aplsp_hard}} (aplsp.north);
        
        \draw[-Implies,double distance=2pt] (rectMinPlus.270) to[]  node[ anchor=center, right] {Thm~\ref{thm:modU_apsp_hard}} (modp.north);
        
        \draw[-Implies,double distance=2pt] (rectMinPlus.310) to[]  node[sloped, anchor=center, above] {Thm.~\ref{thm:min_witness_eq_hard}} (minwitnesseq.north);
        
    \end{tikzpicture}
    \caption{The web of (a subset of) the reductions in this paper. All reductions are $(n^{2+\rho}, n^{2+\rho})$-fine grained reductions, where $\rho$ is such that $\omega(1,\rho,1)=1+2\rho$. The problems in the bounding box are sub $n^{2+\rho}$-equivalent.
    Here, $c_0=\OO(1)$, and $2\le c\le\OO(1)$.}
    \label{fig:APSP_reductions}
\end{figure}


Here we consider the All-Pairs Shortest Paths (APSP) problem in unweighted directed graphs, or more generally in directed graphs with integer weights in $[\pm\ccc]$ with $\ccc=\OO(1)$ and no negative cycles. Zwick \cite{zwickbridge} showed that this problem in $n$-node graphs can be solved in time $\OO(n^{2+\rho})$ time where $\rho$ is such that $\omega(1,\rho,1)=1+2\rho$. For the current best bounds on rectangular matrix multiplication \cite{legallurr}, $\rho$ is roughly $0.529$.

Zwick's algorithm can be viewed as a reduction to rectangular Min-Plus matrix multiplication. The algorithm proceeds in stages, for each $\ell$ from $0$ to $\log_{3/2}(n^{1-\rho})$.

In stage $\ell$, up to logarithmic factors, one needs to compute the Min-Plus product of two matrices $A_\ell$ and $B_\ell$ where $A_\ell$ has dimensions
$n\times n/(3/2)^\ell$ and $B_\ell$ has dimensions $n/(3/2)^\ell \times n$ and both matrices have entries bounded by $(3/2)^\ell$. Intuitively, this computes the pairwise distances that are roughly $(3/2)^\ell$. After stage $\log_{3/2}(n^{1-\rho})$, the algorithm also runs Dijkstra's algorithm\footnote{If there are negative weights, one also needs to run single source shortest paths (SSSP) from a node, as in Johnson's algorithm and then reweight the edges so that they are nonnegative. SSSP can be solved in $\OO((m+n^{1.5})\log^2(c_0)) = \OO(n^2)$ time \cite{van2020bipartite}.} to and from $\OO(n^\rho)$ nodes $S$ sampled randomly and uses $\OO(n^{2+\rho})$ extra time to complete the computation of the distances by considering for every $u,v\in V$, $\min_{s\in S}\{D[u,s]+D[s,v]\}$. This can also be viewed as using the brute-force algorithm to compute the Min-Plus products when $(3/2)^\ell\geq n^{1-\rho}$.

The total running time is within logarithmic factors of
\[n^{2+\rho}+\sum_{\ell=0}^{\log_{3/2}(n^{1-\rho})} \MM(n,n/(3/2)^\ell,n\mid (3/2)^\ell),\]
where $\MM(n_1,n_2,n_3\mid M)$ is the Min-Plus product running time for matrices with entries in $\{0,\ldots,M\}$ and dimensions $n_1\times n_2$ by $n_2\times n_3$.
With the known bounds for Min-Plus product, 
$\MM(n,n^{\tau},n\mid M)\leq \OO(M n^{\omega(1,\tau,1)})$, 
and the running time of Zwick's algorithm becomes $\OO(n^{2+\rho}+n^{1-\rho+\omega(1,\rho,1)})$, which is $\OO(n^{2+\rho})$ when $\omega(1,\rho,1)=1+2\rho$. 

If $\omega=2$, then $\rho$ is $1/2$ and the running time of Zwick's algorithm becomes $\OO(n^{2.5})$. This running time is a seeming barrier for the APSP problem in directed graphs.

In Appendix~\ref{sec:dirapspminplus} we prove the following technical theorem which rephrases Zwick's algorithm \cite{zwickbridge} as a reduction.

\begin{theorem}
\label{thm:APSP_to_Min_Plus}
Let $\rho$ be the solution to $\omega(1,\rho,1)=1+2\rho$.
If the Min-Plus product of an $n\times n^\rho$ matrix by an $n^\rho \times n$ matrix where both matrices have integer entries bounded by $n^{1-\rho}$ (denoted as $\MMp(n,n^\rho,n\mid n^{1-\rho})$) can be computed in $O(n^{2+\rho-\epsilon})$ time for some $\epsilon>0$, then APSP in directed $n$ node graphs with integer edge weights in $[\pm\ccc]$ for $\ccc=\OO(1)$ can be solved in $O(n^{2+\rho-\epsilon'})$ time for $\epsilon'>0$.
\end{theorem}

If $\omega=2$, the above theorem statement becomes:
If the Min-Plus problem of an $n\times \sqrt n$ matrix by a $\sqrt n \times n$ matrix where both matrices have integer entries bounded by $\sqrt n$ can be computed in $O(n^{2.5-\delta})$ time for some $\delta>0$, then APSP in directed $n$ node graphs with integer edge weights in \pmM\ for $\ccc=\OO(1)$ can be solved in $O(n^{2.5-\delta'})$ time for $\delta'>0$.

We will show a reduction in the reverse direction as well, showing that rectangular Min-Plus product with suitably bounded entries can be reduced back to unweighted directed APSP.

\begin{theorem}\label{thm:minplustoapsp}
For any fixed $k\in (0,1)$, $\MMp(n,n^k,n\mid n^{1-k})$  can be reduced in $O(n^2)$ time to APSP in a directed unweighted graph with $O(n)$ vertices.
\end{theorem}

A consequence of  Theorem~\ref{thm:APSP_to_Min_Plus}, and the fact that u-APSP is a special case of  \pmM-APSP for directed graphs without negative cycles, is the following equivalence.

\begin{corollary}
\label{cor:unweightedAPSP_equal_MinPlus}
Let $\rho$ be such that $\omega(1,\rho,1)=1+2\rho$. Then
u-APSP, \pmM-APSP for directed graphs without negative cycles, and $\MMp(n,n^\rho,n\mid n^{1-\rho})$ are sub-$n^{2+\rho}$ fine-grained equivalent for $\ccc=\OO(1)$.
\end{corollary}

In particular, if $\omega=2$, APSP in directed unweighted graphs is sub-$n^{2.5}$ fine-grained equivalent to the Min-Plus problem of an $n\times \sqrt n$ matrix by a $\sqrt n \times n$ matrix where both entries have integer entries bounded by $\sqrt n$.

\begin{proof}[Proof of Theorem~\ref{thm:minplustoapsp}]
Let $A$ be an $n\times n^k$ matrix and let $B$ be an $n^k\times n$ matrix, both with entries in $\{1,\ldots, n^{1-k}\}$.

We will create a directed graph as follows.
Let $I$ be a set of $n$ nodes, which represent the rows of $A$. Let $J$ be a set of $n$ nodes, which represent the columns of $B$.

For every $p\in [n^k]$ corresponding to a column of $A$ (or row of $B$), create a path of $2n^{1-k}+1$ nodes:
\[X(p):=x_{p,n^{1-k}}\rightarrow x_{p,n^{1-k}-1}\rightarrow\ldots\rightarrow x_{p,0} \rightarrow  y_{p,1}\rightarrow y_{p,2}\rightarrow\ldots \rightarrow y_{p,n^{1-k}}.\]

For every $i\in [n]$ and $p\in [n^k]$, consider $t=A[i,p]\in [n^{1-k}]$. Add an edge from $i\in I$ to $x_{p,t}$. Similarly, for every $j\in [n]$ and $p\in [n^k]$, consider $t'=B[p,j]\in [n^{1-k}]$. Add an edge from $y_{p,t'}$ to $j\in J$.

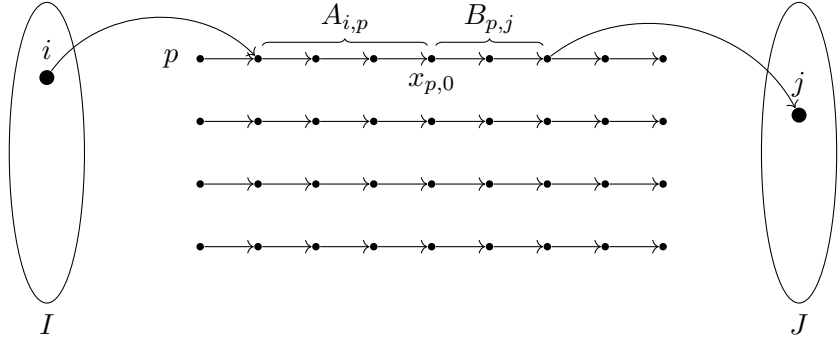
\begin{figure}[ht]
    \centering
    \begin{tikzpicture}
        \node[ellipse, draw,align=left, xshift=-0.5cm, minimum width = 1cm,minimum height = 4cm,label=below:$I$] at (-5, 0) (I) {};
        \node[ellipse, draw,align=left, xshift=-0.5cm, minimum width = 1cm,minimum height = 4cm,label=below:$J$] at (5, 0) (J) {};
        \foreach \i in {1,...,9}
       \foreach \j in {1,...,4}
{
        \pgfmathtruncatemacro{\label}{\i \j};
         \pgfmathtruncatemacro{\ii}{\i -1};
        \pgfmathtruncatemacro{\prevlabel}{\ii \j};
        \node at(\i / 1.3 - 5.5/1.3, \j/1.2-2.5/1.2)  [circle,fill,inner sep=1pt] ({\label}){};
        \ifthenelse{\i>1}{\draw[->,] (\prevlabel) to[]  node[] {} (\label);}{};
}
	\node at(-5.5, 1)  [circle,fill,inner sep=2pt,label=above:$i$] (i){};
	\node at(4.5, 0.5)  [circle,fill,inner sep=2pt,label=above:$j$] (j){};
	\node at(0.5/1.3 -5.5/1.3, 4/1.2-2.5/1.2)  [circle,inner sep=1pt] (p){$p$};
	\node at(5/1.3 -5.5/1.3, 4/1.2-2.5/1.2)  [circle,inner sep=1pt, label=below:{$x_{p,0}$}] (x){};	
	\draw[->, ] (i) to[bend left=50] node[] {} (24);
	\draw[->, ] (74) to[bend left=50] node[] {} (j);
	\draw[decoration={brace,mirror,raise=5pt},decorate]
  (54) -- node[above=6pt] {$A_{i, p}$} (24);
	\draw[decoration={brace,mirror,raise=5pt},decorate]
  (74) -- node[above=6pt] {$B_{p, j}$} (54);
    \end{tikzpicture}
    \caption{Sketch of the construction in proof of Theorem~\ref{thm:minplustoapsp}. 
    For each vertex $i$ and path $p$, we add an edge from $i$ to a vertex on the path $p$ whose distance to the middle point $x_{p, 0}$ on the path is $A_{i, p}$.
    For each path $p$ and vertex $j$, 
     we add an edge from a vertex on the path whose distance from the middle point $x_{p, 0}$ on the path is $B_{p, j}$ to  vertex $j$.}
    \label{fig:APSP_to_MinPlus}
\end{figure}

Now, consider some $i\in [n],p\in [n^k],j\in [n]$ and $A[i,p]+B[p,j]$. If we consider the path consisting of $(i,x_{p,A[i,p]})$, $(y_{p,B[p,j]},j)$ and the subpath of $X(p)$ between $x_{p,A[i,p]}$ and $y_{p,B[p,j]}$, its length is exactly $2+A[i,p]+B[p,j]$. Also, any path from $i$ to $j$ is of this form. 
Thus, the shortest path from $i\in I$ to $j\in J$ in the created graph is exactly of length $2+\min_p \{A[i,p]+B[p,j]\}$, and thus computing APSP in the directed unweighted graph we have created computes the Min-Plus product of $A$ and $B$.

The number of vertices in the graph is $O(n^k \cdot n^{1-k})=O(n)$.
\end{proof}

One consequence of Corollary~\ref{cor:unweightedAPSP_equal_MinPlus} is that u-APSP and computing the predecessor matrix in unweighted directed APSP are also sub-$n^{2+\rho}$ fine-grained equivalent. 
It was known that Zwick's algorithm \cite{zwickbridge} can compute the predecessor matrix for unweighted directed APSP, which can also be viewed as a sub-$n^{2+\rho}$ time reduction from computing the predecessor matrix to $\MMp(n,n^\rho,n\mid n^{1-\rho})$.
Also, if we can compute the  predecessor matrix for the graph constructed in the above proof, we would know which path $X(p)$ the shortest path from $i$ to $j$ uses, which in turn solves $\MMp(n,n^\rho,n\mid n^{1-\rho})$. Thus, computing the predecessor matrix for unweighted directed APSP is sub-$n^{2+\rho}$ fine-grained equivalent to $\MMp(n,n^\rho,n\mid n^{1-\rho})$, and thus also equivalent to u-APSP.

Zwick's algorithm is general enough to apply to some variants of APSP\@. One example is the All-Pairs Longest Paths (APLP) problem in DAGs. To compute APLP in a DAG, we first negate the weight of every edge, then the problem becomes APSP, on which we can directly apply Zwick's algorithm. Therefore, Zwick's algorithm show reductions from u-APLP and \pmM-APLP in DAGs to $\MMp(n,n^\rho,n\mid n^{1-\rho})$.

Perhaps more surprisingly, the other direction of the reduction also holds. Therefore, APLP in DAG and APSP in graphs with weights bounded by $\OO(1)$ are sub-$O(n^{2+\rho})$ equivalent. 



\begin{theorem}
\label{thm:APLPdag_equal_Min_Plus}
Let $\rho$ be such that $\omega(1, \rho, 1) = 1+2\rho$. Then u-APLP in DAGs, \pmM-APLP in DAGs and $\MMp(n,n^\rho,n \mid n^{1-\rho})$  are sub-$n^{2+\rho}$ fine-grained equivalent.
\end{theorem}

The proof of Theorem~\ref{thm:APLPdag_equal_Min_Plus} follows from the same approach and appears in Appendix~\ref{sec:uAPLPminplus}.







All problems shown equivalent to u-APSP above are problems on directed graphs. One natural question is that whether some problems on undirected graphs are also in this equivalence class, or whether we can show some undirected graph problems require $n^{2+\rho-o(1)}$ time if we assume problems in this equivalence class also require $n^{2 + \rho-o(1)}$ time. To answer these questions, we first consider the \RedAPSP{c} problem. 


\begin{theorem}
\label{thm:2redAPSP}
Let $\rho$ be such that $\omega(1, \rho, 1) = 1+2\rho$. \RedAPSP{c} for $2 \le c = \OO(1)$ and  $\MMp(n,n^\rho,n \mid n^{1-\rho})$ are sub-$n^{2+\rho}$ fine-grained equivalent.
\end{theorem}

The proof of Theorem~\ref{thm:2redAPSP} uses a similar graph construction and
 is in Appendix~\ref{sec:thm:2redAPSP}.

By slightly modifying the proof of Theorem~\ref{thm:2redAPSP}, 
we can show conditional hardness for APLSP on undirected graphs where the edge weights are in $\{0, 1\}$. The proof is in Appendix~\ref{sec:thm:2redAPSP}.

\begin{corollary}
\label{cor:aplsp_hard}
Let $\rho$ be such that $\omega(1, \rho, 1) = 1+2\rho$. Suppose $\MMp(n,n^\rho,n \mid n^{1-\rho})$ requires $n^{2+\rho-o(1)}$ time. Then APLSP on undirected graphs where the edge weights can be $\{0, 1\}$ also requires $n^{2+\rho-o(1)}$ time. 
\end{corollary}



Using similar ideas we also show hardness for  \emph{Vertex-Weighted APSP} in undirected graphs, where the vertex weights may be large.
(The current best algorithms for Vertex-Weighted APSP for directed graphs~\cite{ChanSTOC07,YusterSODA09} had running time about
$O(n^{2.85})$; the bound is
$\OO(n^{11/4})$ if $\omega=2$.
No better algorithms were known 
in the undirected graphs---which our conditional lower bound attempts to explain.) The proof is in Appendix~\ref{sec:thm:2redAPSP}.

\begin{corollary}
\label{cor:vertex-weighted-APSP}
Let $\rho$ be such that $\omega(1, \rho, 1) = 1+2\rho$. Suppose $\MMp(n,n^\rho,n \mid n^{1-\rho})$ requires $n^{2+\rho-o(1)}$ time. Then vertex-weighted APSP on undirected graphs where the vertex weights are in $[O(n^{1-\rho})]$ also requires $n^{2+\rho-o(1)}$ time. 
\end{corollary}


The conditional hardness for u-\CountMod{U}
and u-\CountCap{U} 
for any $U \ge 2$
can be proved by combining our graph construction with randomized techniques for a unique variant of Min-Plus product; see Appendix~\ref{sec:thm:modU_apsp_hard}.

\begin{theorem}
\label{thm:modU_apsp_hard}
Let $\rho$ be such that $\omega(1, \rho, 1) = 1+2\rho$. Suppose $\MMp(n,n^\rho,n \mid n^{1-\rho})$ requires $n^{2+\rho-o(1)}$ time (with randomization). Then u-\CountMod{U} 
and u-\CountCap{U} 
for any $U \ge 2$ requires $n^{2+\rho-o(1)}$ time.
\end{theorem}

In Section~\ref{sec:approx}, we will give  an algorithm for approximating APSP with sublinear additive errors. Using the same technique as our reductions from Rectangular Min-Plus product to APSP problems, we can show a conditional lower bound for this problem. 

\begin{theorem}\label{thm:adderr:lb}
Given a directed unweighted graph $G=(V,E)$ with $n$ vertices and a function $f>0$ where $\tfrac{\ell}{f(\ell)}$ is nondecreasing. Suppose we can  approximate the shortest-path distance $D[u,v]$ with
additive error $f(D[u,v])$, for all $u,v\in V$ in $T(n)$ time, then  $\max_{1 \le \ell \le n} \MM\left(n,n/\ell,n\mid \tfrac{\ell}{f(\ell)}\right) \leq O(T(n))$.
\end{theorem}
\begin{proof}
Fix any $1 \le \ell \le n$. First, note that $ \MM\left(n,n/\ell,n\mid \tfrac{\ell}{f(\ell)}\right) = \Theta\left(\MM\left(n,n/\ell,n\mid \tfrac{\ell}{C f(\ell)}\right)\right)$ for any constant $C$. Here, we take $C = 12$ to be a large enough constant. 

Suppose we are given an $n \times n/\ell$ matrix $A$ and an $n/\ell \times n$ matrix $B$, whose entries are positive integers bounded by $\tfrac{\ell}{12 f(\ell)}$, and we want to compute their Min-Plus product $A \star B$.
We use a similar reduction as the one in the proof of Theorem~\ref{thm:minplustoapsp}, but stretching the length of the middle paths. Specifically, we create vertex set $I$ of size $n$, vertex set $J$ of size $n$, and $n/\ell$ paths of the form
$X(p) := x_{p, \frac{\ell}{3f(\ell)}} \leadsto \cdots \leadsto x_{p, 0} \leadsto  y_{p, 0} \leadsto \cdots \leadsto y_{p, \frac{\ell}{3f(\ell)}}.$
From $x_{p, i}$ to $x_{p, i-1}$ and $y_{p, j}$ to $y_{p, j + 1}$, we embed paths of length $6 f(\ell)$; 
from $x_{p, 0}$ to $y_{p,0}$, we embed a path of length $\ell - 2$. Similar to previous reductions, for every $i \in [n] = I$ and $p \in [n/\ell]$, we add an edge from $i$ to $x_{p, A[i, p]}$; for every $j \in [n] = J$ and $p \in [n/\ell]$, we add an edge from $j$ to $y_{p, B[p, j]}$. Then the distance from $i \in I$ to $j \in J$ in this graph equals $\ell + 6f(\ell)  (A\star B)[i, j]$. 

Since $0 \le (A\star B)[i, j] \le \frac{\ell}{6 f(\ell)}$, we must have $\ell \le \ell + 6f(\ell)  (A\star B)[i, j] \le 2\ell$. Since $\frac{\ell}{f(\ell)}$ is nondecreasing, we must have $f(t\ell) \le tf(\ell)$ for any $t \ge 1$, and thus $f(\ell + 6f(\ell)  (A\star B)[i, j]) \le 2f(\ell)$. Therefore, an $f(\ell + 6f(\ell)  (A\star B)[i, j])$-additive approximation of APSP can determine that the distance from $i \in I$ to $j \in J$ is in $\ell + 6f(\ell)  (A\star B)[i, j] \pm 2f(\ell)$, from which we can compute $(A\star B)[i, j]$ easily since $(A\star B)[i, j]$ must be an integer. 
\end{proof}

Finally, we give a reduction from u-APSP to Min Witness Equality, where we are given $n \times n$ integer matrices $A$ and $B$, and are required to compute $\min\{k \in [n]: A[i,k]=B[k, j]\}$ for every pair of $(i, j)$.  Reductions from u-APSP to matrix product problems are considered by Lincoln et al.~\cite{lincoln2020monochromatic}, where they show reductions from u-APSP to the All-Edges Monochromatic Triangle problem and $(\min, \max)$-product problem, but their techniques do not seem to apply to Min Witness Equality.  

The proof of the following theorem is deferred to Appendix~\ref{sec:thm:min_witness_eq_hard}.

\begin{theorem}
\label{thm:min_witness_eq_hard}
Let $\rho$ be such that $\omega(1, \rho, 1) = 1+2\rho$. Suppose $\MMp(n,n^\rho,n \mid n^{1-\rho})$ requires $n^{2+\rho-o(1)}$ time. Then Min Witness Equality requires $n^{2+\rho-o(1)}$ time.
\end{theorem}

\section{Additive Approximation Algorithms for APSP}
\label{sec:approx}

\newcommand{\tD}{\widetilde{D}}


In this section, we
give an algorithm for
approximate APSP with additive errors
in directed unweighted graphs,
to match the lower bound that
we have just proved in
Theorem~\ref{thm:adderr:lb} (ignoring logarithmic
factors).
%
%
Namely, our algorithm achieves running time
$\OO(\max_\ell \MM(n,n/\ell,n\mid \ell^{1-p}))$,
which improves Roditty and Shapira's previous algorithm~\cite{RodittyShapira} with running time $\OO(\max_\ell \min\{n^3/\ell,\MM(n,n/\ell^{1-p},n\mid \ell^{1-p})\})$.  

Let $D[u,v]$ denote the shortest-path distance from $u$ to $v$.

\subparagraph*{Overview.}
The new algorithm is a variation of Zwick's exact u-APSP
algorithm~\cite{zwickbridge}, and is actually simpler than Roditty and
Shapira's algorithm.  The idea is to compute as many as the
shortest-path distances \emph{exactly} as we can in $\OO(n^\omega)$ time in
an initial phase.
In the second phase, we apply rectangular matrix multiplication to
submatrices computed from the first phase, where entries are approximated
by rounding and rescaling.

\subparagraph*{Preliminaries.}
For every $\ell$ that is a power of 3/2,
let $R_\ell\subseteq V$ be a subset of $\OO(n/\ell)$ vertices that hits all shortest paths of length $\ell/2$~\cite{zwickbridge}.
(For example, a random sample works with high probability.)
We may assume that $R_{(3/2)^i} \supseteq R_{(3/2)^{i+1}}$
(because otherwise, we can add $R_{(3/2)^j}$ to $R_{(3/2)^i}$ for all $j>i$
and the size bound would still hold).
For subsets $S_1,S_2\subseteq V$, let $D(S_1,S_2)$ denote the submatrix of $D$
containing the entries for $(u,v)\in S_1\times S_2$.

\subparagraph*{Phase 1.}
We first solve the following subproblem: compute $D[u,v]$ (exactly) for
all $(u,v)\in R_\ell\times V$ with $D[u,v]\le\ell$, and
similarly for
all $(u,v)\in V\times R_\ell$ with $D[u,v]\le\ell$.  

Suppose we have already computed $D[u,v]$ for
all $(u,v)\in R_{2\ell/3}\times V$ with $D[u,v]\le 2\ell/3$,
and similarly for
all $(u,v)\in V\times R_{2\ell/3}$ with $D[u,v]\le 2\ell/3$.

We take the Min-Plus product $D(R_\ell,R_{2\ell/3})\star D(R_{2\ell/3},V)$.
For each $(u,v)\in R_\ell\times V$, if its output entry is smaller than
the current value of $D[u,v]$, we reset $D[u,v]$ to the smaller value.
Similarly, we take the Min-Plus product $D(V,R_{2\ell/3})\star D(R_{2\ell/3},R_\ell)$.
For each $(u,v)\in V\times R_\ell$, if its output entry is smaller than
the current value of $D[u,v]$, we reset $D[u,v]$ to the smaller value.
We reset all entries greater than $\ell$ to $\infty$.

To justify correctness, observe that 
for any shortest path $\pi$ of length between $2\ell/3$ and $\ell$,
the middle $(2\ell/3)/2=\ell/3$ vertices must contain a vertex of $R_{2\ell/3}$, which splits $\pi$ into two subpaths each of length 
at most $\ell/2+\ell/6 \le 2\ell/3$.

We do the above for all $\ell$'s that are powers of $3/2$.
The total cost is
\[
\OO\left(\max_\ell \MM(n/\ell,n/\ell,n\mid \ell)\right)
\le \OO\left(\max_\ell \ell\cdot\MM(n/\ell,n/\ell,n)\right)
\le \OO\left(\max_\ell \ell^2 (n/\ell)^\omega\right) = \OO(n^\omega).
\]

\subparagraph*{Phase 2.}
Next we approximate all shortest-path distances $D[u,v]$ where
$D[u,v]$ is between
$2\ell/3$ and $\ell$, with additive error $O(f(\ell))$ for a given function $f$, as follows:

We compute the Min-Plus product $D(V,R_{2\ell/3})\star D(R_{2\ell/3},V)$, keeping only entries bounded by $O(\ell)$.
As we allow additive error $O(f(\ell))$, we round entries to multiples 
of $f(\ell)$.  This takes $\OO(\MM(n,n/\ell,n\mid \frac{\ell}{f(\ell)}))$ time.

To justify correctness, observe as before that in any shortest path $\pi$ of length between $2\ell/3$ and $\ell$, some vertex in $R_{2\ell/3}$ splits the path into two subpaths of length 
at most $2\ell/3$.

We do the above for all $\ell$'s that are powers of $3/2$.
The total cost is
$
\OO\left(\max_\ell \MM(n,n/\ell,n\mid \tfrac{\ell}{f(\ell)})\right).
$

Standard techniques for generating witnesses for matrix products 
can be applied to
recover the shortest paths (e.g., see \cite{GalilMargalit,zwickbridge}). 

\begin{theorem}
Given a directed unweighted graph $G=(V,E)$ with $n$ vertices and a function $f$ where $\tfrac{\ell}{f(\ell)}$ is nondecreasing,
we can approximate the shortest-path distance $D[u,v]$ with
additive error $O(f(D[u,v]))$ for all $u,v\in V$,
in $\OO\left(\max_\ell \MM(n,n/\ell,n\mid \tfrac{\ell}{f(\ell)})\right)$ time.
\end{theorem}

\subparagraph*{Remark.}
For $f(\ell)=\ell^p$,
we can upper-bound the running time by 
\begin{eqnarray*}
  \OO\left(\max_\ell \MM(n,n/\ell,n\mid \ell^{1-p})\right)
&\le& \OO\left( L^{1-p}\cdot\MMM(n,n/L,n) + n^3/L\right)\\
&\le& \OO\left( L^{1-p} (n^{2+o(1)} + n^\omega/L^{(\omega-2)/(1-\alpha)}) + n^3/L \right)
\end{eqnarray*}
for any choice of $L$,
where $\alpha$ is the rectangular matrix multiplication exponent (satisfying $\omega(1,1,\alpha)=2$).  For example, we can set $L=n^{3-\omega}$, and
for $p > 1- \min\{\frac{\omega-2}{1-\alpha}, \frac{\omega-2}{3-\omega}\}$,  get optimal $\OO(n^\omega)$ running time.  In fact, with the current rectangular
matrix multiplication bounds
we get $\OO(n^{2.373})$ time for $p\ge 0.415\geq (\omega(1,0.373,1)-2\cdot 0.373-1)/(1-0.373)$.
Roditty and Shapira~\cite{RodittyShapira} specifically asked whether there exists $p<1$ for which $\OO(n^\omega)$ time is possible; we have thus answered
their question affirmatively if $\omega>2$.

\subparagraph*{Remark.}
For directed graphs with weights from $[\ccc]$, the running time is
\[\OO\left(\ccc n^\omega + \max_\ell \MM(n,n/\ell,n\mid \ccc \tfrac{\ell}{f(\ell)})\right).\]

\bibliography{ref}

\appendix

\section{Deferred Equivalence and Hardness Proofs}

\subsection{Directed APSP and Rectangular Min-Plus}
\label{sec:dirapspminplus}

Here we prove Theorem \ref{thm:APSP_to_Min_Plus} that reduces u-APSP to rectangular Min-Plus product.

\begin{proof}[Proof of Theorem \ref{thm:APSP_to_Min_Plus}]
Let ALG be an algorithm for $\MMp(n,n^\rho,n\mid n^{1-\rho})$ in $O(n^{2+\rho-\eps})$ time for $\eps>0$.

Let us consider the Min-Plus product running time in each stage of Zwick's algorithm for each $\ell$. Let's pick two parameters $\delta,\delta'\in (0,\rho)$. For any choice of 
 $\ell$ such that $(3/2)^{\ell}\leq n^{1-\rho-\delta}$ for some $\delta>0$, we have that 
$\MM(n,n/(3/2)^\ell,n\mid (3/2)^\ell)$  is bounded from above by $\OO((3/2)^\ell n^{\omega(1,1-\ell/\log_{3/2}(n),1)})$. Since 
 $n^{1-\rho-\delta}/(3/2)^\ell\geq 1$, we can bound $ n^{\omega(1,1-\ell/\log_{3/2}(n),1)})$ from above by $(n^{1-\rho-\delta}/(3/2)^\ell) n^{\omega(1,\rho+\delta,1)}$ by splitting the middle dimension of the matrices to $n^{1-\rho-\delta}/(3/2)^\ell$ pieces and  computing each piece independently. Thus, 
 we get that 
 $\MM(n,n/(3/2)^\ell,n\mid (3/2)^\ell)$ is bounded from above (within polylogarithmic factors) by
 \[(3/2)^\ell n^{\omega(1,1-\ell/\log_{3/2}(n),1)}
 \leq
(3/2)^\ell (n^{1-\rho-\delta}/(3/2)^\ell) n^{\omega(1,\rho+\delta,1)}
=n^{1-\rho-\delta+\omega(1,\rho+\delta,1)}\leq
 n^{2+\rho-\alpha}\] 
for some $\alpha>0$, as $1-r+\omega(1,r,1)$ decreases monotonically as $r$ increases within the interval $[0,1]$.

On the other hand, if we run Dijkstra's algorithm to and from $\OO(n^{\rho-\delta'})$ nodes $S$ and update all pairwise distance estimates for paths going through $S$, this would take $O(n^{2+\rho-\delta'})$ time. This part essentially corresponds to a brute force computation of the  $n\times (n/(3/2)^\ell)\times n$ size Min-Plus products with entries up to $(3/2)^\ell$ for $(3/2)^\ell\geq n^{1-\rho+\delta'}$. 

What remains is to perform $n\times (n/(3/2)^\ell)\times n$ size Min-Plus products with entries up to $(3/2)^\ell$ for $\ell$ such that $(3/2)^\ell\in [n^{1-\rho-\delta},n^{1-\rho+\delta'}]$.

For each $\ell$ such that $n^{1-\rho-\delta}\leq (3/2)^\ell\leq n^{1-\rho}$, the middle dimension of the Min-Plus product is $\leq n^{\rho+\delta}$.
Thus, we can split the middle dimension into $n^\delta$ pieces of size $n^\rho$ and compute $n^\delta$ Min-Plus products of dimension $n\times n^\rho\times n$ with entries up to $(3/2)^\ell\leq n^{1-\rho}$. Computing these $n^\delta$ Min-Plus products using our assumed algorithm ALG takes $O(n^\delta \cdot n^{2+\rho-\eps})$ time. If we set $\delta=c\eps>0$ for some $c \in (0, 1)$, we get (within polylogs) $O(n^{2+\rho-(1-c)\eps})$ time for this step.



Now let's consider $\ell$ such that $n^{1-\rho+\delta'}\geq (3/2)^\ell\geq n^{1-\rho}$; there are $O(\log n)$ such $\ell$.
If we set $N=n^{1+\delta'/(1-\rho)}$, we get that $N^{1-\rho}=n^{1-\rho+\delta'}$ which is the maximum entry in any of the products needed to compute. Also, as $n<N$ and $n^\rho<N^\rho$, we get that
what remains to compute are $O(\log N)$ instances of Min-Plus product of dimension at most $N\times N^\rho\times N$ with entries bounded by $N^{1-\rho}$.


Let's use algorithm ALG to compute these Min-Plus products. The running time of this last step of the reduction becomes within polylogs,
\[n^{(1+\delta'/(1-\rho))(2+\rho-\eps)}=n^{2+\rho-(\eps-\delta'(2+\rho-\eps)/(1-\rho))}.\]
Let us set $\delta'=c\eps (1-\rho)/(2+\rho-\eps)<\rho$.
Then the running time of this step becomes $\OO(n^{2+\rho-(1-c)\eps})$.


We get that APSP in  directed graphs with bounded integer weights can be solved in time (within polylogs) $\OO(n^{2+\rho-\alpha}+n^{2+\rho-(1-c)\eps})$ for any $c \in (0, 1)$.
\end{proof}

\subsection{u-APLP in DAGs and Rectangular Min-Plus}
\label{sec:uAPLPminplus}

\begin{proof}[Proof of Theorem~\ref{thm:APLPdag_equal_Min_Plus}]

A reduction from u-APLP in DAGs to \pmM-APLP in DAGs is trivial, since the former is a special case of the latter. 

To show the reduction from \pmM-APLP in DAGs to $\MMp(n,n^\rho,n \mid n^{1-\rho})$, we first negate the weights of the APLP instance, so the problem becomes APSP. Then we apply  Theorem~\ref{thm:APSP_to_Min_Plus} to complete the reduction. 

It suffices to provide a reduction from $\MMp(n,n^\rho,n \mid n^{1-\rho})$ to u-APLP in DAGs. 

Suppose we are given an $n\times n^\rho$ matrix $A$ and an $n^\rho\times n$ matrix $B$ where both matrices have positive integer entries bounded by $n^{1-\rho}$. First, we create matrices $\bar{A} = n^{1-\rho} + 1 - A$ and $\bar{B} = n^{1-\rho} + 1 - B$. If we compute the Max-Plus product of $\bar{A}$ and $\bar{B}$, we can get the Min-Plus product of $A$ and $B$ by $\min_p \{A[i,p]+B[p, j]\} = 2 + 2n^{1-\rho} - \max_p \{\bar{A}[i,p]+\bar{B}[p, j]\}$. 

Next, we create the same graph from the proof of Theorem~\ref{thm:minplustoapsp} on matrices $\bar{A}$ and $\bar{B}$. The created graph is clearly a DAG. 
We see that now from the APLP of the created graph, we can compute the Max-Plus product of $\bar{A}$ and $\bar{B}$, and thus completing the reduction. 
\end{proof}

\subsection{\RedAPSP{c}, APLSP and Vertex Weighted APSP and Rectangular Min-Plus}
\label{sec:thm:2redAPSP}

\begin{proof}[Proof of Theorem \ref{thm:2redAPSP}]
We first show the reduction from \RedAPSP{c} to $\MMp(n,n^\rho,n \mid n^{1-\rho})$. Since $\MMp(n,n^\rho,n \mid n^{1-\rho})$ is sub-$n^{2+\rho}$ fine-grained equivalent to u-APSP by Corolloary~\ref{cor:unweightedAPSP_equal_MinPlus}, it suffices to show a reduction from \RedAPSP{c} to u-APSP. 

Given an undirected graph $G = (V, R \cup B)$, where $R$ is the set of red edges and $B$ is the set of blue edges. We create a directed graph $G'$ as follows. We copy $V$ to $(c+1)$ parts $V_0, \ldots, V_c$, where a vertex $v \in V$  is copied to $(c+1)$ vertices $v_0 \in V_0, \ldots,  v_c \in V_c$. For each blue edge $e = \{u, v\} \in B$, we add directed edges $(u_i, v_i), (v_i, u_i)$ for each $i \in \{0, \ldots, c\}$. For each red edge $e = \{u, v\} \in R$, we add directed edges  $(u_i, v_{i+1}), (v_i, u_{i+1})$ for $i \in \{0, \ldots , c-1\}$. Now the distance from $u_0$ to $v_i$ in $G'$ is exactly the shortest path distance between $u$ and $v$ that uses exactly $i$ red edges in $G$. Therefore, we can use an algorithm for unweighted directed APSP on $G'$ to compute the pairwise distances in $G'$, then the shortest path length between $u$ and $v$ that uses at most $c$ red edges is $\min_{i=0}^c D_{G'}[u_0, v_i]$.

Now we show the reduction in the other direction. Note that we can reduce \RedAPSP{2} to \RedAPSP{c}  for any $c > 2$ by attaching a length $c-2$ red path $u_{c-2} - u_{c-3} - \cdots u_1 - u$ to any vertex $u$ in the graph. Then the shortest path distance from $u_{c-2}$ to $v$ using at most $c$ red edges in this new graph is exactly the shortest distance from $u$ to $v$ using at most $2$ red edges in the original graph. Thus, it suffices to show the reduction to  \RedAPSP{2}.

This reduction is a modification of the reduction in Theorem~\ref{thm:minplustoapsp}. Let $A$ be an $n \times n^\rho$ matrix and let $B$ be an $n^\rho \times n$ matrix, both with entries in $\{1, \ldots, n^{1-\rho}\}$. 

Let $I$ be a set of $n$ vertices, representing the rows of $A$. Let $J$ be a set of $n$ vertices, representing the columns of $B$. For every $k \in [n^\rho]$ that corresponds to a column of $A$ or row of $B$, we create an undirected path on $2n^{1-\rho}$ vertices, where all edges are blue: 

\[X(k):=x_{k,n^{1-\rho}}- x_{k,n^{1-\rho}-1}-\ldots -  x_{k,0} -  y_{k,1}- y_{k,2}-\ldots - y_{k,n^{1-\rho}}.\]

Finally, for every $i \in [n]$ and $k \in [n^\rho]$, we add a  red edge between $i \in I$ and $x_{k, A[i, k]}$. For every $j \in n$ and $k \in [n^\rho]$, we add a red edge between $j \in J$ and $y_{k, B[k, j]}$.

\begin{figure}[ht]
    \centering
    \begin{tikzpicture}
        \node[ellipse, draw,align=left, xshift=-0.5cm, minimum width = 1cm,minimum height = 4cm,label=below:$I$] at (-5, 0) (I) {};
        \node[ellipse, draw,align=left, xshift=-0.5cm, minimum width = 1cm,minimum height = 4cm,label=below:$J$] at (5, 0) (J) {};
        \foreach \i in {1,...,9}
       \foreach \j in {1,...,4}
{
        \pgfmathtruncatemacro{\label}{\i \j};
         \pgfmathtruncatemacro{\ii}{\i -1};
        \pgfmathtruncatemacro{\prevlabel}{\ii \j};
        \node at(\i / 1.3 - 5.5/1.3, \j/1.2-2.5/1.2)  [circle,fill,inner sep=1pt] ({\label}){};
        \ifthenelse{\i>1}{\draw[-,blue] (\prevlabel) to[]  node[] {} (\label);}{};
}
	\node at(-5.5, 1)  [circle,fill,inner sep=2pt,label=above:$i$] (i){};

	\node at(4.5, 0.5)  [circle,fill,inner sep=2pt,label=above:$j$] (j){};
	\node at(0.5/1.3 -5.5/1.3, 4/1.2-2.5/1.2)  [circle,inner sep=1pt] (p){$p$};
	\node at(5/1.3 -5.5/1.3, 4/1.2-2.5/1.2)  [circle,inner sep=1pt, label=below:{$x_{p,0}$}] (x){};	
	\draw[-, red] (i) to[bend left=50] node[] {} (24);
	\draw[-, red] (74) to[bend left=50] node[] {} (j);
	\draw[decoration={brace,mirror,raise=5pt},decorate]
  (54) -- node[above=6pt] {$A_{i, p}$} (24);
	\draw[decoration={brace,mirror,raise=5pt},decorate]
  (74) -- node[above=6pt] {$B_{p, j}$} (54);
    \end{tikzpicture}
    \caption{Sketch of the reduction in the proof of Theorem~\ref{thm:2redAPSP}. 
    For each vertex $i$ and path $p$, we add a red edge between $i$ and a vertex on the path $p$ whose distance to the middle point $x_{p, 0}$ on the path is $A[i, p]$.
    For each path $p$ and vertex $j$,
     we add a red edge from a vertex on the path whose distance from the middle point $x_{p, 0}$ on the path is $B_{p, j}$ to  vertex $j$.}
    \label{fig:2Red_APSP_to_MinPlus}
\end{figure}
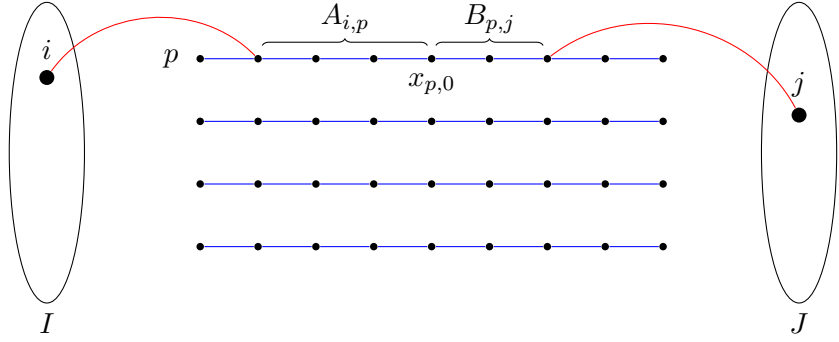

Consider any path from $i \in I$ to $j \in J$ that uses at most $2$ red edges. It must first go to some $X(k)$ using one red edge, going rightwards on $X(k)$ using several blue edges, and finally use another red edge to go to $j$. The length of such a path is $2 + A[i, k] + B[k, j]$. Therefore, the shortest path between $i$ and $j$ that uses at most two red edges have length exactly $2 + \min_k \{A[i, k] + B[k, j]\}$, so calling the \RedAPSP{2} algorithm solves the Min-Plus product instance.
\end{proof}

\begin{proof}[Proof of Corollary \ref{cor:aplsp_hard}]

We consider the reduction from $\MMp(n,n^\rho,n \mid n^{1-\rho})$ to  \RedAPSP{2} in the proof of Theorem~\ref{thm:2redAPSP}. In that reduction, we can replace all red edges with edges of weight $1$, and all blue edges with edges of weight $0$. Then the shortest distance from any vertex $i \in I$ to any vertex $j \in J$ is $2$. Also, the lightest shortest path contains $2 + \min_k \{A[i, k] + B[k, j]\}$ edges. Thus, computing APLSP gives the result of of  an $\MMp(n,n^\rho,n \mid n^{1-\rho})$ instance. 
\end{proof}

\begin{proof}[Proof of Corollary \ref{cor:vertex-weighted-APSP}]
Consider the reduction from $\MMp(n,n^\rho,n \mid n^{1-\rho})$ to  \RedAPSP{2}\@. First, we remove the colors of all edges. For vertices on the paths $X(p)$ for $p \in [n^\rho]$, we set their weights to $1$. For vertices in $I$ and $J$, we set their weight to $2 n^{1-\rho}$. This way, the shortest path from $i \in I$ to $j \in J$ won't visit any other $i'$ or $j'$, so our reduction still follows. 
\end{proof}

\subsection{u-\CountMod{U} and u-\CountCap{U} are Hard from Rectangular Min-Plus}
\label{sec:thm:modU_apsp_hard}

Here we show the conditional hardness for u-\CountMod{U} for any $U \ge 2$.  
The proof for u-\CountCap{U} is similar.

\begin{proof}[Proof of Theorem \ref{thm:modU_apsp_hard}]
First, we define Unique Min-Plus product, where an algorithm is given $A, B$, and is asked to compute $\argmin_{k} \{A[i, k] + B[k, j]\}$. The algorithm only has to be correct on $i, j$ where there exists a unique $k$ that obtains the minimum value; for other $i, j$ pairs, the algorithm is allowed to output any number in $[n^\rho]$. 

We first reduce $\MMp(n,n^\rho,n \mid n^{1-\rho})$ to Unique Min-Plus product of matrices with the same dimensions and entry bounds. 

We perform $\lceil\log(n^\rho)\rceil + 1$ stages, one for each integer $0 \le t \le \lceil \log(n^\rho)\rceil$. During each stage $t$, we will repeat the following for $\Theta(\log n)$ rounds. Let $K$ be the set of column indices of $A$ or the set of row indices of $B$. We independently keep each $k \in K$ with probability $\frac{1}{2^t}$. Thus, we get a submatrix $A'$ of $A$ and a submatrix $B'$ of $B$. We use an algorithm for Unique Min-Plus product to compute $k'[i, j] = \argmin_k \{A'[i, k] + B'[k, j]\}$, and use $A'[i, k'[i, j]] + B'[k'[i, j], j]$ to update our answer for the $(i, j)$-th entry of the Min-Plus product of $A$ and $B$. 

This reduction is correct because when the number of $k$ that achieves minimum value of $A[i, k] + B[k, j]$ is in $[2^t, 2^{t+1}]$, we have a constant probability to keep one unique $k$ if we independently keep every $k \in K$ with probability $\frac{1}{2^t}$. Thus, by repeating the procedure $\Theta(\log n)$ times, we can keep one unique such $k$ in one of the rounds with high probability. 

Then we show a reduction from Unique Min-Plus product to Counting Min-Plus product modulo $U$, where an algorithm needs to compute for every $i, j$, the number of $k$ modulo $U$ such that $A[i,k]+B[k,j] = \min_k \{A[i, k] + B[k, j]\}$. 

For each integer $0 \le t \le \lceil \log(n^\rho)\rceil$, we do the following. For all $k \in [n^\rho]$, if the $t$-th bit in its binary representation is $1$, we duplicate the $k$-th column of $A$ and the $k$-th row of $B$. Then we use an algorithm for Counting Min-Plus product modulo $U$. Suppose there is a unique  $k$ achieving the minimum value of $A[i, k] + B[k, j]$ for $(i, j)$, then the count would be $2 \bmod{U}$ if the $t$-bit in $k$-th binary representation is $1$; otherwise, the count would be $1 \bmod{U}$ (Note that this works even for $U=2$ since the number of witnesses is always at least 1).
Thus, we will be able to completely recover this $k$ for $(i, j)$ if it is unique.

Finally, we reduce Counting Min-Plus product modulo $U$ to u-\CountMod{U}. Note that the Min-Plus product here has two matrices that have dimensions $n \times n^\rho$ and $n^\rho \times n$ respectively and have entries bounded by $n^{1-\rho}$. Thus, we can apply the same reduction from the proof of Theorem~\ref{thm:minplustoapsp}. The number of shortest paths from $i$ to $j$ in that reduction is exactly the number of $k$ that $A[i,k]+B[k,j] = \min_k \{A[i, k] + B[k, j]\}$. 
\end{proof}

\subsection{Min Witness Equality is Hard from Rectangular Min-Plus}
\label{sec:thm:min_witness_eq_hard}

\begin{proof}[Proof of Theorem~\ref{thm:min_witness_eq_hard}]

Suppose we are given integer matrices $A, B$, where $A$ has dimension $n \times n^\rho$, and $B$ has dimension $n^\rho \times n$, and both matrices have entries bounded by $n^{1-\rho}$. We will transform this instance into a Min Witness Equality instance.

Let $A'$ be an $n \times 2n$ matrix. We index the column indices of $A'$ by $\mathcal{V} \times K$, where $\mathcal{V} = [2n^{1-\rho}]$ and $K$ represents the column indices of $A$. Similarly, let $B'$ be a $2n \times n$ matrix, where the row indices of $B$ is also $\mathcal{V} \times K$. We set $A'[i, (v, k)] = A[i, k]$ and $B'[(v, k), j] = v-B[k, j]$. We can easily pad $A'$ and $B'$ to  $2n \times 2n$ square matrices by adding empty rows to $A$ and empty columns to $B$. 

Suppose $A'[i, (v, k)] = B'[(v, k), j]$, then we will have $A[i, k] = v - B[k, j]$, which implies $A[i, k] + B[k, j] = v$. Similarly, if $A[i, k] + B[k, j] = v$, we would have $A'[i, (v, k)] = B'[(v, k), j]$.
Thus, the minimum value of $v$ such that $A'[i, (v, k)] = B'[(v, k), j]$ for some $k$ is the $(i, j)$-th entry of the Min-Plus product of $A$ and $B$. 
Therefore, if we order $\mathcal{V} \times K$ by ordering $\mathcal{V}$ as the primary key, the Min  Witness Equality of $A'$ and $B'$ can easily be used to compute $A \star B$ in $O(n^2)$ time. 

Thus, if we can compute the Min  Witness Equality of $A'$ and $B'$ in $O(n^{2+\rho - \epsilon})$ time for $\epsilon > 0$, we can also compute $A \star B$ in $O(n^{2+\rho - \epsilon})$ time.
\end{proof}





\newcommand{\MMMM}{{\cal M}^{\ast\ast}}

\section{Algorithms for All-Pairs Lightest Shortest Paths}
\label{sec:aplsp}

In this section, we describe algorithms for the following problem,
which includes both All-Pairs Lightest Shortest Paths (APLSP) and Shortest Lightest Paths (APSLP) as special cases:

\begin{problem}
{\bf (Lex$_2$-APSP)} 
We are given a graph $G=(V,E)$ with $n$ vertices,
where each edge $(u,v)\in E$ has a ``primary'' weight
$w_1(u,v)$ and a ``secondary'' weight $w_2(u,v)$.
For every pair of vertices $u,v\in V$, we want to find a path $\pi$ from $u$ to $v$
that minimizes $(\sum_{e\in\pi}w_1(e), \sum_{e\in\pi}w_2(e))$ lexicographically.
\end{problem}

Let $D[u,v]$ be the lexicographical minimum of $(\sum_{e\in\pi}w_1(e), \sum_{e\in\pi}w_2(e))$.
Let $D_1[u,v]$ be the minimum of $\sum_{e\in\pi}w_1(e)$ (the
 shortest-path distance) and let $D_2[u, v]$ be the second coordinate of $D[u,v]$.
APLSP corresponds to the case when all secondary edge weights are 1, whereas
APSLP corresponds to the case when all primary edge weights are 1.




The following lemma, which will be important in the analysis of our Lex$_2$-APSP algorithm, 
bounds the complexity of Min-Plus product of an $n_1\times n_2$ matrix $A$ and an $n_2\times n_3$ matrix $B$ in the case when the finite entries of $A$ come from a small range $[\ell_1]$ (but the finite entries of $B$ may come from a large range $[\ell_2]$).  The bound can be made sensitive to the number $m_2$ of finite entries of $B$ and the number $m_3$ of output entries we want.
The lemma is a variant of \cite[Theorem~3.5]{ChanSTOC07}
(the basic approach originates from Matou\v sek's dominance algorithm~\cite{MatIPL}, but this variant requires some extra ideas). It also generalizes and improves (using rectangular matrix multiplication) Theorem 1.2 in \cite{moreapsp20}. 

\newcommand{\CC}{\widehat{C}}

\begin{lemma}\label{lem:prod:sparse}
$\MM(n_1,n_2,n_3\mid \ell_1,\ell_2)\: =\:
\displaystyle\OO\left(\min_t (\MM(n_1, n_2, n_2n_3/t\mid \ell_1) + tn_1n_3)\right).$

More generally,
$\MM(n_1,n_2,n_3 \mid m_1,m_2,m_3\mid \ell_1,\ell_2) \: =\: \displaystyle\OO\left( \min_t (\MM(n_1, n_2, m_2/t\mid \ell_1) + t m_3)\right).$
\end{lemma}
\begin{proof}

Divide each column of $B$ into groups of $t$ entries by rank: the first group contains the $t$ smallest elements, the second group contains the next
$t$ smallest, etc.\ (ties in ranks can
be broken arbitrarily).
Each column may have at most $t$ leftover entries.  The total number of groups is at most $m_2/t$.

For each $i\in [n_1]$ and $j'\in [m_2/t]$,
let $C[i,j']$ be true iff there exists $k\in [n_2]$ such that $A[i,k]<\infty$ and group $j'$ contains an element with row index $k$.  Computing $C$ reduces to taking a Boolean matrix product and has cost
$O(\MMM(n_1,n_2,m_2/t))$.

For each $i\in [n_1]$ and $j'\in [m_2/t]$, suppose that group $j'$ is part of column $j$ and the maximum element in group $j'$ is $x$;
let $\CC[i,j'] = \min_{k: B[k,j]\in [x,x+\ell_1]} (A[i,k] + B[k,j])$.
Since entries in $A$ are from the range $[\ell_1] \cup \{\infty\}$, and we only keep a size $\ell_1+1$ range of values for matrix $B$, 
computing $\CC$ reduces to taking a Min-Plus product with entries in $[\ell_1]$ (after shifting) and has cost
$O(\MM(n_1,n_2,m_2/t\mid \ell_1))$.

To compute the output entry at each of the $m_3$ positions $(i,j)$, we find the group $j'$ in column $j$ with the smallest rank such that $C[i,j']$ is true.
Let $x$ be the maximum element in group $j'$.  
The answer $\min_k(A[i,k]+B[k,j])$ is at most $x+\ell_1$.
Thus, the answer is defined by an index $k$ that (i)~corresponds to an element in group $j'$, or
(ii)~corresponds to a leftover element in column $j$, or
(iii)~has $B[k,j]\in [x,x+\ell_1]$.
Cases (i) and (ii) can be handled by linear search in $O(t)$ time; case (iii) is handled by looking up $\CC[i,j']$.
The total time to compute $m_3$ output entries is $O(tm_3)$.
%
%
%
\end{proof}

\subsection{\ppM-Lex$_2$-APSP}
\label{sec:aplsp:dir}

Let $\ccc=\OO(1)$.
For directed graphs, Zwick~\cite{ZwickSTOC99} presented a variant
of his u-APSP algorithm that solves \ppM-Lex$_2$-APSP (and thus \ppM-APLSP and \ppM-ALPSP) in 
time $\OO(\max_\ell \MM(n,n/\ell,n\mid \ell^2)) \le \OO(\min_L (L^2 \MMM(n,n/L,n) + n^3/L))$.  This is $O(n^{2.724})$ by the current bounds
on rectangular matrix multiplication~\cite{legallurr} (and is
$\OO(n^{8/3})$ if $\omega=2$).

Chan~\cite{ChanSTOC07} gave a faster algorithm for \pnzM-Lex$_2$-APSP
(and in fact a special case of Vertex-Weighted APSP that includes
\pnzM-Lex$_k$-APSP for an arbitrary constant~$k$)
in time $\OO(n^{(3+\omega)/2})$, which is $O(n^{2.687})$ by the
current matrix multiplication exponent
(and is $\OO(n^{2.5})$ if $\omega=2$). Zwick's algorithm works even when zero primary weights are allowed, but Chan's algorithm does not (part of the difficulty is that 
the secondary distance of a path may be much
larger than the primary distance).
A more general version of Chan's algorithm~\cite{ChanSTOC07} can handle zero primary weights (and \ppM-Lex$_k$-APSP for constant $k$) but has a worse time bound of 
$\OO(n^{(9+\omega)/4})$, which can be slightly reduced using rectangular matrix mutiplication~\cite{YusterSODA09}.

We describe an $O(n^{2.6581})$-time algorithm to solve \ppM-Lex$_2$-APSP
for directed graphs, which can handle zero weights and is faster than Zwick's $O(n^{2.724})$-time algorithm; it is also slightly faster than Chan's algorithm.
The algorithm uses rectangular matrix multiplication (without which the running time would be $\OO(n^{(\omega+3)/2})$).
It should be noted that Chan's previous algorithm can't be easily sped up using rectangular matrix multiplication, besides being inapplicable when there are zero primary weights.

\subparagraph*{Overview.}
The new algorithm can be viewed as an interesting variant of Zwick's u-APSP algorithm~\cite{zwickbridge}.
Zwick's algorithm uses rectangular Min-Plus products of dimensions around $n\times n/\ell$
and $n/\ell\times n$, in geometrically increasing parameter $\ell$.  Our algorithm proceeds 
in two phases.  In both phases, we use the rectangular products of dimensions
around $n/\ell \times n/\ell$ and $n/\ell\times n$.  In the first phase, we consider $\ell$ in increasing order; in the second, we consider $\ell$ in decreasing order.  
In these Min-Plus products, entries of the first matrix in each product come from a small range; this enables us to use
Lemma~\ref{lem:prod:sparse}.

\subparagraph*{Preliminaries.}  
Let $L$ be a parameter to be set later.
Let $\lambda[u,v]$ denote the length of a lexicographical shortest path
between $u$ and $v$.  In this section, the \emph{length} of a path refers
to the number of edges in the path.

For every $\ell$ that is a power of 3/2,
as in Section~\ref{sec:approx},
let $R_\ell\subseteq V$ be a subset of $\OO(n/\ell)$ vertices that hits all shortest paths of 
length $\ell/2$~\cite{ZwickSTOC99,zwickbridge}.
We may assume that $R_{(3/2)^i} \supseteq R_{(3/2)^{i+1}}$
(as before).
Set $R_1=V$.

For $S_1,S_2\subseteq V$, let $D(S_1,S_2)$ denote the submatrix of $D$
containing the entries for $(u,v)\in S_1\times S_2$.

\subparagraph*{Phase 1.}
We first solve the following subproblem for a given $\ell\le L$: compute $D[u,v]$ for
all $(u,v)\in R_\ell\times V$ with $\lambda[u,v]\le\ell$, and
similarly for
all $(u,v)\in V\times R_\ell$ with $\lambda[u,v]\le\ell$.  
(We don't know $\lambda[u,v]$ in advance.  More precisely, if $\lambda[u,v]\le \ell$, the 
computed value should be correct;
otherwise, the computed value is only guaranteed to be an upper bound.)

Suppose we have already computed $D[u,v]$ for
all $(u,v)\in R_{2\ell/3}\times V$ with $\lambda[u,v]\le 2\ell/3$,
and similarly for
all $(u,v)\in V\times R_{2\ell/3}$ with $\lambda[u,v]\le 2\ell/3$.

We take the Min-Plus product $D(R_\ell,R_{2\ell/3})\star D(R_{2\ell/3},V)$
(where elements are compared lexicographically).
For each $(u,v)\in R_\ell\times V$, if its output entry is smaller than
the current value of $D[u,v]$, we reset $D[u,v]$ to the smaller value.
Similarly, we take the Min-Plus product $D(V,R_{2\ell/3})\star D(R_{2\ell/3},R_\ell)$.
For each $(u,v)\in V\times R_\ell$, if its output entry is smaller than
the current value of $D[u,v]$, we reset $D[u,v]$ to the smaller value.
We reset all entries greater than $\ccc\ell$ to $\infty$.

To justify correctness, observe that 
for any shortest path $\pi$ of length between $2\ell/3$ and $\ell$,
the middle $(2\ell/3)/2=\ell/3$ vertices must contain a vertex of $R_{2\ell/3}$, which splits 
$\pi$ into two subpaths each of length 
at most $\ell/2+\ell/6 \le 2\ell/3$.

To take the product,
we map each entry $D[u,v]$ of $D(R_{2\ell/3},V)$ to a number $D_1[u,v]\cdot \ccc\ell + D_2[u,v]\in [\OO(\ell^2)]$. 
It is more efficient to break the product into $\ell$ separate products, by putting
entries of $D(R_\ell,R_{2\ell/3})$ with a common $D_1$ value into one matrix.
Then after shifting, the finite entries of each such matrix are in $[\OO(\ell)]$.
(The entries of $D(R_{2\ell/3},V)$ are still in $[\OO(\ell^2)]$.)
Hence, the computation takes time $\OO(\ell\cdot \MM(n/\ell,n/\ell,n\mid 
\ell,\ell^2))$.

We do the above for all $\ell\le L$ that are powers of $3/2$ (in increasing order).

\subparagraph*{Phase 2.}
Next we solve the following subproblem for a given $\ell\le L$: compute $D[u,v]$ for
all $(u,v)\in R_{2\ell/3}\times V$ with $\lambda[u,v]\le L$.  

Suppose we have already computed $D[u,v]$ for
all $(u,v)\in R_{\ell}\times V$ with $\lambda[u,v] \le L$.

We take the Min-Plus product $D(R_{2\ell/3},R_{\ell})\star D(R_{\ell},V)$, keeping only 
entries bounded by $\OO(\ell)$ in the first matrix and $\OO(L)$ in the second matrix.
For each $(u,v)\in V\times R_\ell$, if its output entry is smaller than
the current value of $D[u,v]$, we reset $D[u,v]$ to the smaller value.

To justify correctness, recall that 
for $(u,v)\in R_{2\ell/3}\times V$,
if $\lambda[u,v]<2\ell/3$, then $D[u,v]$ is already computed in Phase 1.
On the other hand, in any shortest path $\pi$ of length between $2\ell/3$ and~$L$, the first $\ell/2$ vertices of the path must contain a vertex of $R_\ell$.

To take the product,
we map each entry $D[u,v]$ of $D(R_{2\ell/3},V)$ to a number $D_1[u,v]\cdot \ccc L + D_2[u,v]\in [\OO(\ell L)]$. 
As before, it is better to perform $\ell$ separate products, by putting
entries of $D(R_{2\ell/3},R_{\ell})$ with a common $D_1$ value into one matrix.
Then after shifting, the finite entries of each such matrix are in $[\OO(\ell)]$.
(The entries of $D(R_{2\ell/3},V)$ are still in $[\OO(\ell L)]$.)
Hence, the computation takes time $\OO(\ell\cdot \MM(n/\ell,n/\ell,n\mid 
\ell,\ell L))$.

We do the above for all $\ell\le L$ that are powers of $3/2$ (in decreasing order).

\subparagraph*{Last step.}
By the end of Phase 2 (when $\ell$ reaches 1), we have computed $D[u,v]$ for all
$(u,v)$ with $\lambda[u,v]\le L$.
To finish, we compute $D[u,v]$ for all $(u,v)$ with $\lambda[u,v]>L$, as follows:

We run Dijkstra's algorithm $O(|R_L|)$ times to compute $D[u,v]$ for all $(u,v)\in R_L\times 
V$ and for all $(u,v)\in V\times R_L$.  This takes
$O(|R_L|n^2)=\OO(n^3/L)$ time.
We then compute $D(V,R_L)\star D(R_L,V)$ by brute force in
$O(|R_L|n^2)=\OO(n^3/L)$ time.

Correctness follows since every shortest path of length more than $L$
must pass through a vertex in $R_L$.

As before, standard techniques for generating witnesses for matrix products 
can be applied to
recover the shortest paths \cite{GalilMargalit,zwickbridge}. 

\subparagraph*{Total time.}
The cost of Phase~2 dominates the cost of Phase~1.  By Lemma~\ref{lem:prod:sparse},
the total cost is 
\begin{eqnarray*}
\lefteqn{\OO\left(\max_{\ell\le L}  
\ell\cdot \MM(n/\ell,n/\ell,n\mid \ell,\ell L)
 \,+\, n^3/L\right)}\\
 &\le& 
\OO\left(\max_{\ell\le L} \ell\cdot \min_t \left(\MM(n/\ell,n/\ell, n^2/(\ell t)\mid \ell) + tn^2/\ell\right) \,+\, n^3/L\right).
\end{eqnarray*}
We set $t=n/L$ and obtain
\[\OO(\max_{\ell\le L}\ell^2\cdot \MMM(n/\ell,n/\ell, Ln/\ell) \,+\, n^3/L).\]
Intuitively, the maximum occurs when $\ell=1$, and so 
we should choose $L$ to minimize $\OO(\MMM(n,n,Ln) + n^3/L)$.
With the current bounds on rectangular matrix multiplication \cite{legallurr}, we choose $L=n^{0.342}$ and get running time $O(n^{2.6581})$.
(Formally, we can verify this time bound using
the convexity of the function $2x + \omega(1-x,1-x,1.342-x)$.)


\begin{theorem}
\ppM-Lex$_2$-APSP (and thus \ppM-APLSP and \ppM-APSLP) 
can be solved in $O(n^{2.6581})$ time for any  $\ccc=\OO(1)$.
%
\end{theorem}

\subparagraph*{Remarks.}
Without rectangular matrix multiplication, the above still
gives a time bound of $\OO(Ln^\omega + n^3/L)$, yielding $\OO(n^{(3+\omega)/2})$.

The same algorithm works even with negative weights (i.e., for \pmM-Lex$_2$-APSP),
like Zwick's previous algorithm~\cite{ZwickSTOC99}, assuming no negative cycles.

In Appendix~\ref{sec:aplsp:dir2}, we describe an alternative algorithm that has the same running time, though it does not allow zero primary edge weights (or negative weights).

\subsection{Undirected \pnzM-Lex$_2$-APSP}
\label{sec:aplsp:undir}

A natural question is whether APLSP or APSLP is easier for undirected graphs.
We now describe a faster $O(n^{2.58})$-time algorithm for \pnzM-Lex$_2$-APSP for undirected graphs.  Zero primary weights are not allowed, but zero secondary weights are.  (In particular, the algorithm can solve \ppM-APSLP, when all primary weights are 1.)

\subparagraph*{Overview.}
We follow an idea of Aingworth et al.~\cite{AingworthCIM99}, to divide into
two cases: when the source vertex has high degree or low degree.
For high-degree vertices, there exists a small dominating set, and so
these vertices can be covered by a small number of ``clusters''; sources
in the same cluster are close together, and so distances from one fixed
source give us good approximation to distances from other sources 
in the same cluster, by the triangle inequality (since the graph is undirected).  On the other hand, for low-degree vertices, the relevant
subgraph is sparse, which enables faster algorithms.  Originally, Aingworth et al.'s approach was intended for the design of approximation algorithms (with $O(1)$ additive error for unweighted graphs).
We will adapt it to find \emph{exact} shortest paths.
(Chan~\cite{ChanSODA06} previously had also applied Aingworth et al.'s approach to 
exact APSP, but the goal there was in logarithmic-factor speedup, which
was quite different.)  In order to handle the high-degree case for Lex$_2$-APSP, we need further ideas to use approximate primary shortest-path distances to compute exact lexicographical shortest-path distances; in particular, 
we will need 
Min-Plus products on secondary distances (as revealed in the proof
of Lemma~\ref{lem:multisource} below).
The combination of Aingworth et al.'s approach with matrix multiplication appears new, and interesting in our opinion.

\subparagraph*{Preliminaries.}
We first compute $D_1[u,v]$ for all $(u,v)$ by running a known \ppM-APSP algorithm on the primary distances in $O(n^\omega)$ time~\cite{AlonGalilMargalit,Seidel}.

Assume that we have already
computed $D[u,v]$ for all $(u,v)$ with $D_1[u,v]\le 2\ell/3$ for a given $\ell$.
We want to compute $D[u,v]$ for all $(u,v)$ with $D_1[u,v]\le \ell$.

\newcommand{\DDD}[2]{\DD{#1}(#2)}
Define $\DD{\ell}[u,v]=D_2[u,v]$ if $D_1[u,v]=\ell$, and $\DD{\ell}[u,v]=\infty$ otherwise.
For subsets $S_1,S_2\subseteq V$, let $\DDD{\ell}{S_1,S_2}$ denote the submatrix of $\DD{\ell}$
containing the entries for $(u,v)\in S_1\times S_2$.

\begin{lemma}\label{lem:multisource}
Let $G=(V,E)$ be an undirected graph with edge weights in
$[\ccc]-\{0\}$.  Assume that we have already
computed $D[u,v]$ for all $(u,v)$ with $D_1[u,v]\le 2\ell/3$.
Given a set $S$ of vertices that are within primary distance $c=\OO(1)$ from
each other, we can compute $D[u,v]$ for all $u\in S$ and $v\in V$ with $D_1[u,v]\le\ell$ in $O(\MM(|S|,n/\ell,n\mid \ell))$ total time.
\end{lemma}
\begin{proof}
Fix $s\in S$.
Let $V_i=\{v\in V: D_1[s,v]\in i\pm c\}$.
Note that $\sum_i |V_i| = \OO(n)$.
Also note that if $u\in S$ and $D_1[u,v]=i$, then we must have $v\in V_i$
(by the triangle inequality, because the graph is undirected).

Pick an index $m\in [0.4\ell,0.6\ell]$ with $|V_{m-\ccc}\cup\cdots\cup V_m|=\OO(n/\ell)$.

For $i\le m$, we have already computed $\DDD{i}{S,V_i}$.

For $i=m+1,\ldots,\ell$, we will compute $\DDD{i}{S,V_i}$ as follows:
For each $\D\in [\ccc]$, we take
the Min-Plus product
$\DDD{m-\D}{S,V_{m-\D}}\star \DDD{i-m+\D}{V_{m-\D},V_i}$.
Note that $\DDD{i-m+\D}{V_{m-\D},V_i}$ is already known, since $i-m+\D < 2\ell/3$.
We take the minimum over all $\D\in [\ccc]$ for those 
$(u,v)\in S\times V_i$ with
$D_1[u,v]=i$.

Instead of doing the product individually for each $i$,
it is more efficient to combine all the matrices
$\DDD{i-m+\D}{V_{m-\D},V_i}$ over all $i>m$.
This gives a single matrix (per $\D$) with $|V_{m-\D}|=\OO(n/\ell)$ rows and
$\sum_{i>m} |V_i| = \OO(n)$ columns.
So, the entire product can be computed in
$O(\MM(|S|,n/\ell,n\mid \ell))$ time.
\end{proof}

\newcommand{\VH}{V_{\mbox{\scriptsize\rm high}}}
\newcommand{\VL}{V_{\mbox{\scriptsize\rm low}}}

Let $L$ be a parameter to be set later.
Let $\VH$ be the set of all vertices of degree more than $n/L$,
and $\VL$ be the set of all vertices of degree at most $n/L$.

\subparagraph*{Phase 1.}
We will first compute $D[u,v]$ for all $u\in \VH$ and $v\in V$ with $D_1[u,v]\le\ell$, as follows:

Let $X\subseteq V$ be a dominating set for $\VH$ of size $\OO(L)$, such that
every vertex in $\VH$ is in the (closed) neighborhood of some vertex in $X$.
Such a dominating set can be constructed (for example, by the standard
greedy algorithm) in $\OO(n^2)$ time~\cite{AingworthCIM99}.

Let $X=\{x_1, x_2, \ldots, x_{\OO(L)}\}$. For each $x_i \in X$, we divide $N(x_i) \setminus \left(\bigcup_{j < i} N(x_j)\right)$ -- its neighborhood excluding previous neighborhoods --  into groups of size $O(n/L)$.  The total number of groups is $\OO(L)$,
and the groups cover all vertices in $\VH$.  For each such group, we apply Lemma~\ref{lem:multisource} (with $c=2\ccc$).
The total time is $\OO(L\cdot \MM(n/L,n/\ell,n\mid \ell))$.

\subparagraph*{Phase 2.}
Next, for each $u\in \VL$, we will compute $D[u,v]$ for all $v\in V$ with
$D_1[u,v]\le\ell$, as follows:

Define a graph $G_u$ containing all edges $(x,y)$ with $x\in \VL$ or $y\in \VL$;
for each $z\in \VH$, we add an extra edge $(u,z)$ with weight
$D[u,z]$, which has been computed in Phase~1.
Then the lexicographical shortest-path distance from $u$ to $v$ in $G_u$
matches the lexicographical shortest-path distance in $G$,
because if $\langle u_1,\ldots,u_k\rangle$ is
a lexicographical shortest path in $G$  
with $u_1=u$, and $i$ is the largest index with $u_i\in \VH$ (set $i=1$ if
none exists),
then $\langle u_1,u_i,\dots,u_k\rangle$ is 
a path in $G_u$.
We run Dijkstra's algorithm on $G_u$ from the source $u$.
Since $G_u$ has $O(n^2/L)$ edges, this takes $\OO(n^2/L)$ time per $u$.
The total over all $u$ is $\OO(n^3/L)$.

As before, standard techniques for generating witnesses for matrix products 
can be applied to
recover the shortest paths~\cite{GalilMargalit,zwickbridge}. 

\subparagraph*{Total time.}
We do the above for all $\ell$'s that are powers of $3/2$.
The overall cost is
\begin{eqnarray*}
\lefteqn{\OO\left(\max_\ell L \cdot \MM(n/L,n/\ell,n\mid \ell) \,+\, n^3/L\right)}\\
&\le& \OO\left(\max_\ell L\cdot \min\left\{ n^3/(L\ell),\ \ell\cdot \MMM(n/L,n/\ell,n)\right\} \,+\, n^3/L\right)\\
&=& \OO\left(\max_{\ell\le L} L\ell \cdot \MMM(n/L,n/\ell,n) \,+\, n^3/L\right)
\ =\ \OO(L^2\cdot \MMM(n/L,n/L,n) \,+\, n^3/L).
\end{eqnarray*}

With the current bounds on rectangular matrix multiplication, we choose $L=n^{0.4206}$ and get running time $O(n^{2.5794})$.



\begin{theorem}
\pnzM-Lex$_2$-APSP (and thus -APLSP and -APSLP) for undirected graphs
 can be solved in 
$O(n^{2.5794})$ time for any $\ccc=\OO(1)$.
\end{theorem}

\subparagraph*{Remarks.}
Without rectangular matrix multiplication, the above still
gives a time bound of $\OO(L^3 (n/L)^\omega + n^3/L)$, yielding $\OO(n^{2+1/(4-\omega)})$. 

One could adapt the algorithm to solve Undirected \pnzM-Lex$_k$-APSP for a larger constant $k$, but the running time appears worse than
the bound $\OO(n^{(3+\omega)/2})$ by Chan~\cite{ChanSTOC07}
(because of the need to compute a Min-Plus product between matrices with larger entries in Lemma~\ref{lem:multisource}).  


\section{Algorithms for APSP Counting}
\label{sec:counting}

In this section, we describe algorithms for the following problem:

\begin{problem}
{\bf ($\#$APSP)}
Given a graph $G=(V,E)$,
we want to count the number $C[u,v]$ of shortest paths from $u$ to $v$
for every pair of vertices $u,v\in V$.
\end{problem}

This problem has direct applications to \emph{Betweenness Centrality}~\cite{Brandes,Brandes08}, a problem that, given a graph and a vertex $v$, asks to compute the quantity $\sum_{s,t\in V\setminus \{v\}} C_v[s,t]/C[s,t]$, where 
$C_v[s,t]$ is the number of shortest paths from $s$ to $t$ that go through $v$ (i.e., $C_v[s,t]=C[s,v]\cdot C[v,t]$ if $D[s,t]=D[s,v]+D[v,t]$, and $C_v[s,t]=0$ otherwise). The known algorithms for this problem in $m$-edge, $n$-node graphs are said to run in $\tilde{O}(mn)$ time~\cite{Brandes}, however these algorithms seem to work in a model where arbitrarily sized integers can be added in constant time.
Note that the counts could be exponentially large in the worst case, and hence if one takes the size of the integers into account, the running time of all prior algorithms is actually $\tilde{O}(mn^2)$ (and there are simple examples on which it actually runs in $\Omega(mn^2)$ time for any choice of $m$).

We will consider several variants of the $\#$APSP  problem:
\CountCap{U} (computing $\min\{C[u,v],U\}$ for a given value $U$),
\CountMod{U} (computing $C[u,v]\bmod U$),
\CountApx{U} (computing a $(1+1/U)$-factor approximation), and the original exact version. Note that
\CountCap{U} reduces to \CountApx{O(U)} easily,
and \CountCap{U} also reduces to \CountMod{\OO(nU)} (by using random moduli in $\widetilde{\Theta}(nU)$). 

Let $D[u,v]$ denote the shortest-path distance from $u$ to $v$.


\subsection{u-\CountCap{U}}
\label{sec:counting:directedcap}

We solve u-\CountCap{U} for directed graphs by a variant of
Zwick's u-APSP algorithm~\cite{zwickbridge}.

\newcommand{\funnystar}{\bullet}
Given a pair $(A,A')$ of $n_1\times n_2$ matrices and a pair $(B,B')$ of  $n_2\times n_3$ matrices $B$, define a new pair of $n_1\times n_3$ matrices $(C,C')=(A,A')\funnystar (B,B')$, where
$C[i,j]=\min_k \{ A[i,k]+B[k,j] \}$
and $C'[i,j]=\sum_{k:A[i,k]+B[k,j]=C[i,j]} A'[i,k] B'[k,j]$.
Assuming that the finite entries of $A$ and $B$ are in $[\ell]$ and the entries of $A'$ and $B'$ are in $[U]$, we can reduce
this ``funny'' product $(C,C')$ to 
a standard matrix product by
mapping $(A[i,j],A'[i,j])$ to $A'[i,j]\cdot M^{A[i,j]}$
and $(B[i,j],B'[i,j])$ to $B'[i,j]\cdot M^{B[i,j]}$,
for a sufficiently large $M=\text{poly}(n_2,U)$.
Thus, the computation time is $\OO(\ell\cdot \MMM(n_1,n_2,n_3))$, ignoring $\log (n_1n_2n_3U)$ factors.  Alternatively, the funny product can be computed trivially in $\OO(n_1n_2n_3)$ time.

Let $\Pi$ contain up to $U$ shortest paths for every pair of vertices.  (This collection $\Pi$ of paths is not given to us but helps in the analysis.)
For every $\ell$ that is a power of 2, 
let $R_\ell\subseteq V$ be a random subset of $c(n/\ell)\log(nU)$ vertices for a sufficiently large constant $c$.  With high probability, we have the following properties (the second of which follows from a Chernoff bound):
(i)~every subpath of a path in $\Pi$ of length $\ell/2$ contains a vertex of $R_\ell$, and (ii)~every path in $\Pi$ of length at most $\ell$ contains at most $b =\Theta(\log(nU))$ vertices of $R_\ell$.
We may assume that $R_{2^i} \supseteq R_{2^{i+1}}$.
Take $R_{2n}=\emptyset$.

We first compute $D[u,v]$ for all $u,v\in V$ by running Zwick's u-APSP algorithm.

Let $C_{\ell'}[u,v]$ denote the number of shortest
paths from $u$ to $v$ (of distance $D[u,v]$) where all intermediate
vertices are in $V-R_{\ell'}$; it is set to $U$ if the number exceeds the cap~$U$.  Ultimately, we want $C_{2n}[u,v]$.

Suppose we have already computed the counts $C_{\ell'}[u,v]$ for all $u,v\in V$ with
$D[u,v]\le \ell/2$, for all $\ell'$'s that are powers of 2.  
We want to compute the counts $C_{\ell'}[u,v]$ for all $u,v\in V$ with
$D[u,v]\in (\ell/2,\ell]$ for all $\ell'$'s that are powers of 2.
 
\newcommand{\hC}{\widehat{C}}
For subsets $S_1,S_2\subseteq V$, let 
$D(S_1,S_2)$ and $C_{\ell'}(S_1,S_2)$ denote the submatrix of $D$ and $C_{\ell'}$ respectively
containing the entries for $(u,v)\in S_1\times S_2$.
Let $\hC_{\ell'}(S_1,S_2)$ denote the matrix pair $(D(S_1,S_2), C_{\ell'}(S_1,S_2))$.

For each $j=0,\ldots,b$,
for $\ell'> \ell$,
we compute the funny product $\hC_{\ell}(V,R_{\ell}-R_{\ell'})\funnystar
\hC_{\ell}(R_{\ell}-R_{\ell'},R_{\ell}-R_{\ell'})^j\funnystar
\hC_{\ell}(R_{\ell}-R_{\ell'},V)$ (the $j$-th power here is with respect to the operator $\funnystar$).  The counts we want are the entries of the second matrix in the output for all $(u,v)$ with $D[u,v]\in (\ell/2,\ell]$, summed over all $j$.  When computing the product, we can place a cap of $U$ in all intermediate second matrices.
The running time is $\OO(\min\{\ell\cdot \MMM(n,n/\ell,n),\, n^3/\ell\})$ (since $j$ is polylogarithmic).  For $\ell'\le\ell$, the counts are set to 0.

To justify correctness, consider a shortest path $\pi$ in $\Pi$ of length in $(\ell/2,\ell]$ where all intermediate vertices are from $V-R_{\ell'}$.
Then $\pi$ must pass through at least one vertex of $R_\ell$ (this implies that $\ell'>\ell$) and at most $b$ vertices of $R_\ell$; furthermore, any subpath $\pi'$ between two consecutive vertices in $R_\ell$, or between the start vertex and the first vertex in $R_\ell$, or between the last vertex in $R_\ell$ and the last vertex, have at most $\ell/2$ vertices.
Note that $\pi'$ has all intermediate vertices from $V-R_{\ell'}-R_\ell=V-R_\ell$.
Thus, each such path $\pi'$ will be counted by one of the above products exactly once.

We do the above for all $\ell$'s that are powers of 2.
The total time is $\OO(\max_\ell \min\{\ell\cdot \MMM(n,n/\ell,n), n^3/\ell\})
$, which can be bounded above by $\OO(L\cdot \MMM(n,n/L,n) + n^3/L)$ for any choice of $L$.  This bound is similar to that of Zwick's u-APSP algorithm.
We set $L=n^{0.471}$.


\begin{theorem}
u-\CountCap{U} can be solved in $O(n^{2.529}\log^{O(1)}U)$ time with high probability.
\end{theorem}

\subparagraph*{Remarks.}
The above is not quite a reduction to directed unweighted APSP, since
$\ell\cdot \MMM(n,n/\ell,n)$ is not quite the same as
$\MM(n,n/\ell,n\mid \ell)$.
For the case of $U=\OO(1)$, we can turn the algorithm into a (randomized) reduction:

\newcommand{\ffunnystar}{\otimes}
Define another product $C''=A\ffunnystar B$, where
$C''[i,j]$ is the number of $k$'s for which $A[i,k]+B[k,j]=C[i,j]$, where
$C[i,j]=\min_k \{A[i,k]+B[k,j]\}$.
If the number $C''[i,j]$ is capped at $U$,
we can reduce $\ffunnystar$ to $O(U^2)$ Min-Plus products:
take a random partition $R_1,\ldots,R_{cU^2}$, and for each $\ell\in [cU^2]$,
compute $\min_{k\in R_\ell} \{A[i,k]+B[k,j]\}$.
Set $C''[i,j]$ to be the number of such witnesses $k$
with $\min_{k\in R_\ell}A[i,k]+B[k,j]=C[i,j]$.
For each $C[i, j]$, this is correct with probability greater than $1/2$ for a sufficiently large constant $c$ (we can repeat logarithmically many times to lower the error probability).

We can reduce $\funnystar$ to polylogarithmically many $\ffunnystar$ products:
$C'[i,j]=\sum_{p,q\in [\log U]} 2^{p+q}\sum_k (A_p\ffunnystar B_q)[i,j]$,
where $A_p[i,k]=A[i,k]$ if the $p$-th bit of $A'[i,k]$ is 1, and 
$A_p[i,k]=\infty$ otherwise, and
$B_q[k,j]=B[k,j]$ if the $q$-th bit of $B'[k,j]$ is 1, and 
$B_q[k,j]=\infty$ otherwise.

Section~\ref{sec:reductions} gives a reduction in the other direction, and thus u-\CountCap{U}
for $2\le U\le \OO(1)$ is equivalent to u-APSP,
completing the proof of Theorem~\ref{thm:apspcapu}.

\subsection{Undirected u-\CountCap{U} and u-\CountMod{U}}
\label{sec:counting:undirected}


For undirected unweighted graphs, Seidel's algorithm with Zwick's modification \cite{zwickpersonal} is as follows:
given a Boolean adjacency matrix $A$ of $G$, compute $A^2=A\vee (A\cdot A)$ here $\vee$ is componentwise OR and $\cdot$ is the Boolean matrix product.  Let $G^2$ be the graph defined by $A^2$. Recursively compute the pairwise distances $d_2(u,v)$ in $G^2$. The base case is when the diameter is $1$ which is easy to handle (we work on connected graphs, since we can work on each connected component separately, and the diameter roughly halves in each step.)

Then, we note that the distance $D[u,v]$ is odd iff $v$ has a neighbor $x$ such that $d_2(u,v)\equiv d_2(u,x)+1\pmod 3$ as for every neighbor $w$ of $v$, $d_2(u,w)\in \{d_2(u,v)-1,d_2(u,v),d_2(u,v)+1\}$.

So we define for each $j \in \{0, 1, 2\}$, $B_j[u,v]=1$ if $d_2(u,v)\equiv j\pmod 3$ and $B_j[u,v]=0$ otherwise. Then we multiply for each $j \in \{0, 1, 2\}$, $C_j=B_j \cdot A$, and for every $u,v$, we compute $j= (d_2(u,v)-1) \bmod 3$ and if $C_j[u,v]=1$, we conclude that $D[u,v]$ is odd. Otherwise, we conclude that $D[u,v]$ is even. For every even $D[u,v]$, we can compute it by setting $D[u,v]=2d_2(u,v)$, and for odd $D[u,v]$ we compute it by setting $D[u,v]=2d_2(u,v)-1$.

Now we want to compute the counts of the paths together with Seidel's approach. Given $A$ which is now viewed as an integer matrix where $A[i,j]$ is the multiplicity of edge $(i,j)$ ($0$ if no edge), we compute $\bar{A}^2=A+A^2$ over the integers. This defines $G^2$, a graph with new multiplicity adjacency matrix $\bar{A}^2$. 
Recurse on $G^2$, obtaining the distances $d_2(u,v)$ and shortest paths counts $c_2(u,v)$ in $G^2$. As in Seidel's algorithm we can compute $D[u,v]$ (the distances in $G$) from $d_2(u,v)$. Now we also want to compute the shortest path counts in $G$. If $D(u,v)$ is even, then $d_2(u,v)=D[u,v]/2$ and all $u-v$ paths in $G^2$ correspond to $u-v$ paths in $G$, and $C[u,v]=c_2(u,v)$, so we can just set these counts because we know which distances are even. 

If $D[u,v]$ is odd, on the other hand, for every predecessor $x$ on a shortest $u-v$ path: (1) the number of $u-x$ paths is $C[u,x]=c_2(u,x)$, and (2) the number of $u-v$ paths going through $x$ is $c_2(u,x)\cdot A[x,v]$.
From the above version of Seidel's algorithm we know that  $x$ is a predecessor of $v$ on a $u-v$ shortest path iff $D[u,x]\equiv D[u,v]-1\pmod 3$. So we can compute the count $C[u,v]$ as follows. For each $j\in \{0, 1, 2\}$, let $D_j$ be the matrix defined as $D_j[u,x]=c_2(u,x)$ if $D[u,x] \equiv j\pmod 3$ and $0$ otherwise. Then set $X_j$ to be the product $D_j A$. Now, for every $u,v$ for which $D[u,v]$ is odd, let $j=(D[u,v]-1)\bmod 3$, and look at $X_j[u,v]$. By our discussion, this will be the sum over all $x$ such that $x$ is a neighbor of $v$ and $D[u,x]\equiv D[u,v]-1\pmod 3$, of $c[u,x]\cdot A[x,v]$ which is exactly the number of shortest paths from $u$ to $v$, as the graph is undirected. Hence we can return $D$ and $C$.

If we do computations modulo $U$, the integers will be bounded by $U$ and the runtime will be $\tilde{O}(n^\omega \log U)$.

If we set all counts that are greater than $U$ as $U$, we will get the APSP counts capped at~$U$.

\begin{theorem}
u-\CountMod{U} and u-\CountCap{U} in undirected graphs can be computed in $\tilde{O}(n^\omega \log U)$ time. 
\end{theorem}

\subsection{Undirected u-\CountApx{U}}
\label{sec:counting:undirectedapprox}

We can solve u-\CountApx{U} for undirected graphs, interestingly, by 
adapting the APLSP algorithm in Appendix~\ref{sec:aplsp:undir}.
The only main change is to replace the Min-Plus products on the secondary distances
with standard matrix products on the counts, with approximation factor $1+O(1/U)$.  

Let $\MMMM(n_1,n_2,n_3\mid \ell)$ be the time to compute the standard product
of an $n_1\times n_2$ matrix with an $n_2\times n_3$ matrix where all finite entries are from $[2^\ell]$.  It is known that $\MMMM(n_1,n_2,n_3\mid \ell)
= \OO(\ell\cdot \MMM(n_1,n_2,n_3)\})$.
(The naive $\OO(n_1n_2n_3)$ bound still holds when computing the product approximately, with factor $1+O(1/U)$, ignoring $\polylog U$ factors.)
In the analysis, we just replace $\MM$
with $\MMMM$, since the number of paths of length $\ell$ is bounded by $2^{\OO(\ell)}$.

In the unweighted case, we can simplify by setting $\D=0$.  Dijkstra's algorithm can be generalized for approximate counting.
The approximation factor may increase to $(1+O(1/U))^{O(n)}$, which is acceptable after readjusting $U$ by an $O(n)$ factor.

\begin{theorem}\label{thm:countapx:undir}
u-\CountApx{U} for undirected graphs can be solved in 
$O(n^{2.5794}\log^{O(1)}U)$ time.
\end{theorem}

\subsection{Exact u-$\#$APSP}
\label{sec:counting:exact}

For exact counts that could be exponentially large,
we will describe a combinatorial $\OO(n^3)$-time algorithm to solve u-$\#$APSP for directed unweighted graphs, in the standard word RAM model (with $(\log n)$-bit words).
The algorithm can be viewed
as a special case of the method in Appendix~\ref{sec:aplsp:undir}
with $L=1$ (no matrix multiplication and no dominating sets are required, and the method turns out
to work for the directed case). 

We first compute $D[u,v]$ for all $u,v\in V$ in $O(n^3)$ time
by known APSP algorithms. There are of course faster APSP algorithms for directed unweighted graphs, but we use the slower $O(n^3)$ time algorithm to keep the whole algorithm combinatorial. 

Assume we have already
computed $C[u,v]$ for all $u,v$ with $D[u,v]\le 2\ell/3$ for a given $\ell$.
Fix a source vertex $s\in V$.
We will compute $C[s,v]$ for all $v$ with $D[s,v]\le \ell$, as follows:

Let $V_i=\{v\in V: D[s,v]=i\}$.
Note that $\sum_i |V_i| = n$, so there exist an index $m\in [0.4\ell,0.6\ell]$ with $|V_m|=O(n/\ell)$.

For $i\le m$, we have already computed $C[s,v]$ for all $v\in V_i$.

For $i=m+1,\ldots,\ell$, we compute $C[s,v]$ for all $v\in V_i$ by setting
$C[s,v]=\sum_{u\in V_m: D[u,v]=i-m} C[s,u]\cdot C[u,v]$.
Note that $C[s,u]$ and $C[u,v]$ have been computed from the previous
iteration, since $i-m < 2\ell/3$.
The total number of arithmetic operations 
is $O(\sum_i |V_i|\cdot |V_m|)=O(n^2/\ell)$.
Since the counts are bounded by $O(n^\ell)$ and are $\OO(\ell)$-bit numbers,
the total cost is $\OO(n^2/\ell\cdot \ell)=\OO(n^2)$.




We do this for every source $s\in V$.  The overall cost is $\OO(n^3)$.

We do the above for all $\ell$'s that are powers of $3/2$.
The final time bound is $\OO(n^3)$.

\begin{theorem}
u-\#APSP can be solved in 
$\OO(n^3)$ time.
\end{theorem}

\subparagraph*{Remarks.}
This is worst-case optimal up to polylogarithmic factors, as
the total number of bits in the answers could be $\Omega(n^3)$.

Recall the Betweenness Centrality of a vertex $v$ is defined as 
$\text{BC}(v)=\sum_{s,t\neq v} C_v[s,t]/C[s,t]$ where $C_v[s, t]$ is the number of shortest paths  between $s$ and $t$ that go through $v$.
As an immediate corollary, we can compute
the Betweenness Centrality of a given vertex exactly in a directed unweighted graph
in $\OO(n^3)$ time
(or approximately in an undirected unweighted graph with factor $1+1/U$ in $O(n^{2.5794}\log^{O(1)}U)$ time
by Theorem~\ref{thm:countapx:undir}).

\begin{corollary}
The betweenness centrality of a vertex can be computed in $\OO(n^3)$ time in a directed unweighted graph. Furthermore, it can be approximated with factor $1+1/U$ in an undirected unweighted graph  in $O(n^{2.5794}\log^{O(1)}U)$ time.
\end{corollary}

In Appendix~\ref{sec:alternative}, we give more algorithms for
u-\CountMod{U} and
u-\CountApx{U} for directed graphs.

Many of the algorithms in this section (and in Appendix~\ref{sec:alternative}) can be extended to graphs with weights in $[\ccc]-\{0\}$
for any $\ccc=\OO(1)$
(for example, the $\OO(n^3)$ algorithm still works with
$([\ccc]-\{0\})$-$\#$APSP after some modifications), but in the interest of simplicity, we will not go into the details.  

\section{Alternative \pnzM-Lex$_2$-APSP Algorithm and its Applications}
\label{sec:alternative}

\subsection{\pnzM-Lex$_2$-APSP}

\label{sec:aplsp:dir2}

We describe an alternative $O(n^{2.6581})$-time algorithm for \pnzM-Lex$_2$-APSP
for directed graphs without zero edge weights. 
Although it has identical running time as the (more general) algorithm in Appendix~\ref{sec:aplsp:dir}, the approach has applications to certain versions of $\#$APSP, as we will see later.

\subparagraph*{Overview.}
The general plan consists of two phases.
In the first phase, we compute all shortest-path distances for which
the primary distances are (roughly) powers of 2; this is done by ``repeated squaring'' with the Min-Plus product (involving the secondary weights).
In the second phase, we compute all shortest-path distances for
all primary distances divisible by $2^i$, for $i=\log n$ down to 1,
by more Min-Plus products using the matrices computed in the first phase.
There were previous APSP algorithms that follow a similar plan for 
undirected graphs (e.g., Shoshan and Zwick's algorithm~\cite{shoshanzwick}), but the difficulty in our problem is that there is
no clear way to reduce the magnitude of the numbers in the Min-Plus products
involving the secondary weights.  We suggest an extra, simple idea:
pick a random number $x$; then the number of entries with primary distance equal to $x$ will be small (this is a simplified statement---in Lemma~\ref{lem:gamma} below, the random numbers we use are obtained by multiplying distances with a random scaling factor~$\gamma$).  This way, we
have reduced the 
\emph{sparsity} of the matrices in the Min-Plus products, and can then apply
Lemma~\ref{lem:prod:sparse} (specifically, the second bound that is sensitive to the number of finite entries of the matrices).  

\subparagraph*{Preliminaries.}
We first compute $D_1[u,v]$ for all $(u,v)$, by running Zwick's \ppM-APSP algorithm on the primary weights.

\begin{lemma}\label{lem:gamma}
There exists a real $\gamma\in[1,2]$ such that for each $i$,
the number of $(u,v)$ with $D_1[u,v]\in \{\floor{\gamma j 2^i}: j \in \mathbb{Z}^+\} \pm O(1) $ is $\OO(n^2/2^i)$. 
Such a $\gamma$ can be
computed in $\OO(n^2)$ time.
\end{lemma}
\begin{proof}
Pick a random $\gamma\in \{1,1+\frac1n, \ldots,2\}$.
For any fixed $a\in [c_0n]$ and $i$, we have $\Pr [\exists j\ge 1: \floor{\gamma j2^i}=a] = \Pr[\exists j\ge 1: \frac a{j2^i}\le\gamma < \frac{a+1}{j2^i}] = O(\sum_{j=1}^{\floor{c_0n/2^i}} \frac 1{j2^i}) = O((1/2^i)\log n)$.

Thus, for each fixed $(u,v)$, the probability that $D_1[u,v] \in \{\floor{\gamma j 2^i}: j \in \mathbb{Z}^+\} \pm O(1)$ is $O((1/2^i)\log n)$.
So, the expected number of $(u,v)$ with such $D_1[u,v]$ is $O((n^2/2^i)\log n)$.
By Markov's inequality, the probability that the number exceeds $c(n^2/2^i)\log^2n$ is $O(1/(c\log n))$ for a fixed $i$.  We take the union bound over all $i\le \log(c_0n)$.

To derandomize, we first count the number of $(u,v)$ with $D_1[u,v]=a$, 
for every $a\in [c_0n]$, in $O(n^2)$ time.  Afterwards, for each $\gamma\in \{1,1+\frac1n, \ldots,2\}$ and each $i\le\log(c_0n)$, we can check whether $\gamma$ satisfies the property in $O(n)$ time.
\end{proof}

Define $\DD{\ell}[u,v]=D_2[u,v]$ if $D_1[u,v]=\ell$, and $\infty$ otherwise.
We are now ready to present our algorithm.

\subparagraph*{Phase 1.}
For $i=0,\ldots,\log n$,
we will compute $\DD{\floor{\gamma 2^i}+b}$ for all $b\in\{-4c_0,\ldots,4c_0\}$ as follows:

For the base case $i=0$, we can compute $\DD{b}$ for all $b=O(1)$
by $O(1)$ number of Min-Plus products with finite entries bounded by $O(1)$,
in $\OO(n^\omega)$ time.

Fix arbitrary $i\ge 1$.
For each $\D\in [c_0]$,
we will compute selected entries of the Min-Plus product $\DD{\floor{\gamma 2^{i-1}}+\floor{b/2}-\D+e_i}\star 
\DD{\floor{\gamma 2^{i-1}}+\ceil{b/2}+\D}$,
where $e_i = \floor{\gamma 2^i} - 2\floor{\gamma 2^{i-1}}\in \{0,1\}$.
Note that $\DD{\floor{\gamma 2^{i-1}}+\floor{b/2}-\D+e_i}$ and 
$\DD{\floor{\gamma 2^{i-1}}+\ceil{b/2}+\D}$ have been computed in the previous iteration.
We compute only those output entries for
$(u,v)$ with $D_1[u,v]=\floor{\gamma 2^i}+b$.
For each such $(u,v)$, we take the minimum of the output entries over all $\D\in [c_0]$. 

By Lemma~\ref{lem:gamma}, $\DD{\floor{\gamma 2^{i-1}}\pm O(1)}$ and
$\DD{\floor{\gamma 2^i}\pm O(1)}$ have
$\OO(n^2/2^i)$ finite entries, all from $[O(2^i)]$.  
So, the $\OO(n^2/2^i)$  output entries we want can be computed 
in time
\[\OO(\MM(n,n,n\mid n^2/2^i, n^2/2^i, n^2/2^i\mid 2^i,2^i)).\]

\subparagraph*{Phase 2.}
For $i=\log n,\ldots,0$, 
we compute $\DD{\floor{\gamma j2^i}+b}$
for all $j\in [n/2^i]$ and for all $b\in\{-4c_0,\ldots,4c_0\}$, as follows:

For $j\le 1$, the answers have already been computed.

For $j$ even, the answers have been computed in the previous iteration.

Suppose $j>1$ is odd.
For each $\D\in [c_0]$, we compute selected entries of the Min-Plus product
$\DD{\floor{\gamma (j-1)2^i} +\floor{b/2}-\D +e_{ij}}
\star \DD{\floor{\gamma 2^i} +\ceil{b/2}+\D}$,
where $e_{ij}=\floor{\gamma j2^i} - \floor{\gamma (j-1)2^i}-
\floor{\gamma 2^i} \in \{0,1\}$.
Note that $\DD{\floor{\gamma (j-1)2^i} +\floor{b/2}-\D}$ has been
computed in the previous iteration, and
$\DD{\floor{\gamma 2^i} +\ceil{b/2}+\D + e_{ij}}$ has been computed
in Phase~1.
We compute only those output entries for
$(u,v)$ with $D_1[u,v]=\floor{\gamma j2^i}+b$.
For each such $(u,v)$, we take the minimum of the output entries over all $\D\in [c_0]$. 

Instead of doing the product individually for each $j$,
it is more efficient to combine all the matrices
$\DD{\floor{\gamma (j-1)2^i} +\floor{b/2}-\D +e_{ij}}$ over 
 $j\in [n/2^i]$.  This gives a single matrix (per $i,b,\Delta$)
with $O((n/2^i)n)$ rows and $n$ columns; by Lemma~\ref{lem:gamma},
this matrix has $\OO(n^2/2^i)$ finite entries,
all from $[O(n)]$.
The second matrix $\DD{\floor{\gamma 2^i} +\ceil{b/2}+\D}$
has $\OO(n^2/2^i)$ finite entries,
all from $[O(2^i)]$.
So, the $\OO(n^2/2^i)$ output entries we want can be computed
in time
\[\OO(\MM(n^2/2^i,n,n\mid n^2/2^i, n^2/2^i, n^2/2^i\mid n,2^i)).\]

By the end of Phase 2, we have computed $\DD{\ell}$ for all $\ell$. 
Standard techniques for generating witnesses for matrix products can be applied to
recover the paths corresponding to the lexicographical shortest-path distances~\cite{GalilMargalit}. 

\subparagraph*{Total time.}
The cost of Phase~2 dominates.
By Lemma~\ref{lem:prod:sparse} (with the two matrices reversed), the total cost is bounded by
\begin{eqnarray*}
&&\OO\left(\max_\ell \MM(n^2/\ell,n,n\mid n^2/\ell, n^2/\ell, n^2/\ell\mid n,\ell)\right)\\
&\le&\OO\left(\max_\ell \min\left\{n^3/\ell,\ \min_t(\ell\cdot\MMM(n^2/(\ell t),n, n) + tn^2/\ell)\right\}\right).
\end{eqnarray*}
We choose $t=\ell n/L$ for some parameter $L$ to be determined.
The cost is at most 
\[\OO\left(\max_{\ell\le L}\ell\cdot \MMM(Ln/\ell^2,n, n) + n^3/L\right).\]
The maximum occurs when $\ell=1$, and so we should choose $L$ to minimize $\OO(\MMM(Ln,n,n) + n^3/L)$.
With the current bounds on rectangular matrix multiplication, we choose $L=n^{0.342}$ and get $O(n^{2.6581})$ running time.

\IGNORE{
More precisely, 
Le Gall~\cite{LeGall} showed that $\omega(1.3,1,1) < 2.624703$
and $\omega(1.4,1,1) < 2.711707$.
By convexity, $\omega(1.34026,1,1) <
\frac{1.4-1.34026}{1.4-1.3} 2.624703 + \frac{1.34026-1.3}{1.4-1.3} 2.711707
< 2.6598$.
Furthermore, by convexity, for all $x\in [0,0.34026]$,
$x+\omega(1.34026-2x,1,1)\le 
\max\{\omega(1.34026,1,1),\, 0.34026+\omega(0.65974,1,1)\} < 2.6598$.
}

\subsection{u-\CountMod{U}}


Zwick's APSP algorithm does not seem generalizable to \CountMod{U}: when there are exponentially many shortest paths, the existence of a small hitting set is unclear.

We solve \CountMod{U} for directed unweighted graphs, interestingly by 
adapting the APLSP algorithm in Appendix~\ref{sec:aplsp:dir2}.
The only main change is to replace the Min-Plus products on the secondary distances
with standard matrix products on the counts (over $\mathbb{Z}_U$).
In the unweighted case, we can simplify by setting $\D=0$.
In the analysis,
let $\MMM_U(n_1,n_2,n_3\mid m_1,m_2,m_3)$
be the time to compute $m_3$ given entries
of the standard matrix product
of an $n_1\times n_2$ matrix with
an $n_2\times n_3$ matrix, where
all entries are integers in $[U]$.
By an analog to Lemma~\ref{lem:prod:sparse},
\[\MMM_U(n_1,n_2,n_3 \mid m_1,m_2,m_3) \:=\: \displaystyle\OO\left( \min_t (\MMM(n_1, n_2, m_2/t) + t m_3)\right),\]
where the $\OO$ notation may hide $\log^{O(1)}(n_1n_2n_3U)$ factors.
The cost of the entire algorithm is no worse than in Appendix~\ref{sec:aplsp:dir2}, ignoring $\log^{O(1)}U$ factors.
We thus obtain:

\begin{theorem}
u-\CountMod{U} can be solved in 
$O(n^{2.6581}\log^{O(1)}U)$ time.
\end{theorem}


\subsection{u-\CountApx{U}}


We can also solve u-\CountApx{U} for directed unweighted graphs by 
adapting the APLSP algorithm in Appendix~\ref{sec:aplsp:dir2}.
As before, the main change is to replace the Min-Plus products on the secondary distances
with standard matrix products on the counts, but this time with approximation.

Let $\MMMM(n_1,n_2,n_3\mid \ell)$ be the time to compute the standard product
of an $n_1\times n_2$ matrix with an $n_2\times n_3$ matrix where all 
finite entries are from $[2^\ell]$.  It is known that $\MMMM(n_1,n_2,n_3\mid \ell)
\le \OO(\ell\cdot \MMM(n_1,n_2,n_3))$.

Let $\MMMM_U(n_1,n_2,n_3 \mid m_1,m_2,m_3\mid \ell_1,\ell_2)$ be the time to compute $m_3$ given entries of
the standard product of an $n_1\times n_2$ matrix $A$ with an $n_2\times n_3$ matrix $B$, where $A$ has at most $m_1$ finite entries, all from $[2^{\ell_1}]$,
and $B$ has at most $m_2$ finite entries, all from $[2^{\ell_2}]$, allowing approximation factor $1+1/U$.

\begin{lemma}\label{lem:prod:sparse:apx}
\[\MMMM_U(n_1,n_2,n_3 \mid m_1,m_2,m_3\mid \ell_1,\ell_2) \ =\ \displaystyle\OO\left( \min_t (\MMMM(n_1, n_2, m_2/t\mid \ell_1) + t m_3)\right),\]
where the $\OO$ notation may hide $\log^{O(1)}(n_1n_2n_3U)$ factors.
\end{lemma}
\begin{proof}
The proof is similar to that of Lemma~\ref{lem:prod:sparse}. 

For each $i\in [n_1]$ and $j'\in [m_2/t]$,
let $C[i,j']$ be true iff there exists $k\in [n_2]$ such that $A[i,k]>0$ and  group $j'$ contains an element with row index $k$.  Computing $C$ reduces to taking a Boolean matrix product and has cost
$O(\MMM(n_1,n_2,m_2/t))$.

For each $i\in [n_1]$ and $j'\in [m_2/t]$, suppose that group $j'$ is part of row $j$ and the minimum element in group $j'$ is $x$;
let $\CC[i,j'] = \sum_{k: B[k,j]\in [x/(2^{\ell_1}n_2U), x]}
A[i,k] \cdot B[k,j]$.
Computing $\CC$ reduces to multiplying two matrices with reals entries in $[1,2^{\ell_1}n_2U]$ after rescaling, or integers 
in $[2^{\ell_1}n_2U^2]$ after
rounding (since we allow approximation factor $1+1/U$), and has cost
$\OO(\MMMM(n_1,n_2,m_2/t\mid \ell_1))$.

To compute the output entry at position $(i,j)\in [n_1]\times [n_3]$, we find the group $j'$ in column $j$ with the largest rank  such that $C[i,j']$ is true.
Let $x$ be the minimum element in group $j'$.  
We compute $\sum_k A[i,k]\cdot B[k,j]$ over every index $k$ that
(i)~corresponds to an element in group $j'$, or
(ii)~corresponds to a leftover element in column $j$, or
(iii)~has $B[k,j]\in [x/(2^{\ell_1}n_2U),x]$.
Note that the sum over $k$ with $B[k,j] < x/(2^{\ell_1}n_2 U)$ is bounded by $x/U$,
and can be omitted when approximating with factor $1+O(1/U)$.
Terms associated with cases (i) and (ii) can be handled by in $O(t)$ time; terms associated with case (iii) can be handled by looking up $\CC[i,j']$.
The total time to compute $m_3$ output entries is
$O(tm_3)$.
\end{proof}

The remaining analysis is similar, noting that the number of paths of length $\ell$ is bounded by $2^{\OO(\ell)}$. The approximation factor may increase to
$(1+O(1/U))^{O(n)}$, which is acceptable after readjusting $U$ by an $O(n)$ factor.

\begin{theorem}
u-\CountApx{U} can be solved in 
$O(n^{2.6581}\log^{O(1)}U)$ time.
\end{theorem}

\section{\RedAPSP{1}}

We showed that \RedAPSP{c} is equivalent to u-APSP for directed graphs when $2 \le c = \OO(1)$, so it requires $\Omega(n^{2.5})$ time unless there is a breakthrough for u-APSP. In this section, we show that \RedAPSP{1} is actually an easier problem.

\begin{theorem}
There is an $\tilde{O}(n^\omega)$ time algorithm for \RedAPSP{1} in unweighted undirected graphs.
\end{theorem}
\begin{proof}
We adapt Seidel's algorithm.
In Seidel's algorithm, given the adjacency matrix $A$, we compute $A^2=A\vee (A\cdot A)$ which represents the adjacency matrix of vertices with distances at most $2$. Then we solve APSP in the graph with adjacency matrix $A^2$ recursively, and use that result to compute APSP in the original graph. 

In the \RedAPSP{1} problem, the graph contains more information than an adjacency matrix. Let $R$ be the adjacency matrix for red edges, and $B$ be the adjacency matrix for blue edges. We will define $(R, B)^2$, which basically represents the graph where we combine two adjacent edges into a single edge. Let $(R, B)^2 = (R \vee (R \cdot B) \vee (B \cdot R), B \vee (B \cdot B))$, and let $D_2[u, v]$ denote the shortest paths distance between $u$ and $v$ using at most $1$ red edge in the graph represented by $(R, B)^2$. We can compute $D_2$ recursively. 
In the base case, the distance between any two vertices is either $1$ or $\infty$. The base case is achieved after roughly $O(\log n)$ recursive calls, since after each call, the distances of reachable pairs are roughly halved. 

For simplicity, we call paths that use at most one red edges ``valid''.

Suppose  we are given $D_2$, we can compute the real distances $D$ as follows. If $D_2[u, v] = \infty$, then so is $D[u, v]$. In the following, we handle the case when $D_2[u, v] < \infty$. 
We first compute some value $\bar{D}[u, v]$, which is defined as follows: if $u$ has a blue neighbor $x$ (i.e. the edge connecting $u$ and $x$ is blue) such that $D_2[u, v]\equiv D_2[x, v]+1 \pmod{3}$, we set $\bar{D}[u, v] = 2D_2[u, v]-1$; otherwise, we set $\bar{D}[u, v] = 2D_2[u, v]$. Clearly $\bar{D}$ can be computed in $O(n^\omega)$ time using Boolean matrix multiplication. 

\begin{claim}
\label{cl:1redapsp}
$\bar{D}[u, v] \ge D[u, v]$. Also, if there exists a shortest valid  path from $u$ to $v$ that uses a blue edge as its first edge, then $\bar{D}[u, v] = D[u, v]$.
\end{claim}
\begin{proof}

First, notice that $D_2[u, v] = \lceil D[u, v] / 2\rceil$ for any pairs of $u, v$. By triangle inequality, for any blue neighbor $x$ of $u$, $|D[u, v] - D[x, v]| \le 1$, which implies that $|D_2[u, v] - D_2[x, v]| \le 1$. Therefore, we can ignore the mod $3$ condition in the construction of $\bar{D}$. 

The only case when $\bar{D}[u, v] < D[u, v]$ could happen is when $u$ has a blue neighbor $x$ such that $D_2[u, v]=D_2[x, v]+1$ and $D[u, v] = 2D_2[u, v]$. However, in this case, $$D[u, v] \le 1+D[x, v] \le 1+2D_2[x, v] = 1+2(D_2[u, v]-1) = D[u, v] - 1, $$
which is a contradiction. 

Now suppose there exists a shortest valid  path from $u$ to $v$ that goes to a blue neighbor $x$ first. If $D[u, v]$ is odd, then $D_2[x, v] < D_2[u, v]$, so we will  set $\bar{D}[u, v] = 2D_2[u, v] - 1$, which equals $D[u, v]$. Similarly, if $D[u, v]$ is even, then there cannot exist
$x$ such that $D_2[u, v] \equiv D_2[x, v]+1\pmod{3}$, so we do the correct thing by setting $\bar{D}[u, v] = 2D_2[u, v]$. 
\end{proof}

We can compute $D$ from $\bar{D}$ as follows. If there exists an edge between $u$ and $v$, then we set $D[u, v] = 1$; otherwise, we set $D[u, v] = \min\{\bar{D}[u, v], \bar{D}[v, u]\}$. It is correct since if the length of the shortest valid path is at least $2$, either the first edge is blue or the last edge is blue. 

\end{proof} 

\section{Simpler Unweighted $[\ccc]$-APSP}\label{sec:simple:undir}

Shoshan and Zwick~\cite{shoshanzwick} gave an algorithm for the standard $[\ccc]$-APSP problem for undirected graphs running in $\OO(\ccc n^\omega)$ time.
In this section, we describe a simple alternative based on our two-phase approach
from Sections~\ref{sec:approx} and~\ref{sec:aplsp:dir}.

\begin{lemma}\label{lem:minplus:variant}
Let $A$ be an $n_1\times n_2$ matrix with integer entries from $[\ell]\cup\{\infty\}$,
and let $B$ be an
$n_2\times n_3$ matrix with (possibly large) integer entries,
satisfying the following property: for every $i,j,k,k'$, 
\begin{equation}\label{tri:ineq}
B[k,j]\le A[i,k]+A[i,k']+B[k',j].
\end{equation}
Then the Min-Plus product of $A$ and $B$ can be computed
in $O(\MM(n_1,n_2,n_3\mid \ell))$ time.
\end{lemma}
\begin{proof}
\newcommand{\BB}{\widehat{B}}
Define $\BB^{(t)}[i,j]=(B[i,j]+t\ell)\bmod{6\ell}$.
For each $t=0,\ldots,5$,
compute the Min-Plus product $A\star\BB^{(t)}$ in $O(\MM(n_1,n_2,n_3\mid \ell))$ time.

Consider an $(i,j)\in [n_1]\times [n_3]$.
Let $k_0$ be any index such that $A[i,k_0]$ is finite.
Let $t\in\{0,\ldots,5\}$ with $((B[k_0,j]+t\ell)\bmod{6\ell})\in [2\ell,3\ell]$.
For all $k$ such that $A[i,k]$ is finite,
(\ref{tri:ineq}) implies that $|B[k_0,j]-B[k,j]|\le 2\ell$,
and so $B[k,j] - \BB^{(t)}[k,j] = B[k_0,j] - \BB^{(t)}[k_0,j]$.
Thus, the index $k$ minimizing $A[i,k]+B[k,j]$ is the same as
the index $k$ minimizing $A[i,k]+\BB^{(t)}[k,j]$---which we have already found.
\end{proof}

Let $\lambda[u,v]$ denote the length of a shortest path
between $u$ and $v$, where the \emph{length} of a path refers
to the number of edges in the path.

For every $\ell$ that is a power of 3/2,
as in Section~\ref{sec:approx},
let $R_\ell\subseteq V$ be a subset of $\OO(n/\ell)$ vertices that hits all shortest paths of 
length $\ell/2$~\cite{zwickbridge}.
We may assume that $R_{(3/2)^i} \supseteq R_{(3/2)^{i+1}}$
(as before).
Set $R_1=V$.

For $S_1,S_2\subseteq V$, let $D(S_1,S_2)$ denote the submatrix of $D$
containing the entries for $(u,v)\in S_1\times S_2$.

\subparagraph*{Phase 1.}
We first solve the following subproblem for a given $\ell$: compute $D[u,v]$ for
all $(u,v)\in R_\ell\times V$ with $\lambda[u,v]\le\ell$.
(We don't know $\lambda[u,v]$ in advance.  More precisely, if $\lambda[u,v]\le \ell$, the 
computed value should be correct;
otherwise, the computed value is only guaranteed to be an upper bound.)

Suppose we have already computed $D[u,v]$ for
all $(u,v)\in R_{2\ell/3}\times V$ with $\lambda[u,v]\le 2\ell/3$,
and thus, by symmetry, for
all $(u,v)\in V\times R_{2\ell/3}$ with $\lambda[u,v]\le 2\ell/3$.

We take the Min-Plus product $D(R_\ell,R_{2\ell/3})\star D(R_{2\ell/3},V)$.
For each $(u,v)\in R_\ell\times V$, if its output entry is smaller than
the current value of $D[u,v]$, we reset $D[u,v]$ to the smaller value.
We reset all entries greater than $\ccc\ell$ to $\infty$.

To justify correctness, observe that 
for any shortest path $\pi$ of length between $2\ell/3$ and $\ell$,
the middle $(2\ell/3)/2=\ell/3$ vertices must contain a vertex of $R_{2\ell/3}$, which splits 
$\pi$ into two subpaths each of length 
at most $\ell/2+\ell/6 \le 2\ell/3$.

The computation takes time $\OO(\MM(n/\ell,n/\ell,n\mid 
\ccc\ell))$.
We do the above for all $\ell\le L$ that are powers of $3/2$ (in increasing order).

\subparagraph*{Phase 2.}
Next we solve the following subproblem for a given $\ell$: compute $D[u,v]$ for
all $(u,v)\in R_{2\ell/3}\times V$ (with no restrictions on $\lambda[u,v]$).  

Suppose we have already computed $D[u,v]$ for
all $(u,v)\in R_{\ell}\times V$, and thus, by symmetry,
for all $(u,v)\in V\times R_\ell$.

We take the Min-Plus product $D(R_{2\ell/3},R_{\ell})\star D(R_{\ell},V)$, keeping only 
entries bounded by $\ccc\ell$ in the first matrix.
For each $(u,v)\in V\times R_\ell$, if its output entry is smaller than
the current value of $D[u,v]$, we reset $D[u,v]$ to the smaller value.

To justify correctness, recall that 
for $(u,v)\in R_{2\ell/3}\times V$,
if $\lambda[u,v]\le 2\ell/3$, then $D[u,v]$ is already computed in Phase 1.
On the other hand, in any shortest path $\pi$ of length at least $2\ell/3$, the first $\ell/2$ vertices of the path must contain a vertex of $R_\ell$.

Observe that in the above product, (\ref{tri:ineq}) is satisfied due to the
triangle inequality, since the graph is undirected and the matrices 
$D(R_{2\ell/3},R_{\ell})$ and $D(R_{\ell},V)$ are true shortest path distances
(by the induction hypothesis).
Hence, by Lemma~\ref{lem:minplus:variant},
the computation takes time $\OO(\MM(n/\ell,n/\ell,n\mid 
\ccc\ell))$.
We do the above for all $\ell$ that are powers of $3/2$ (in decreasing order).

As before, standard techniques for generating witnesses for matrix products 
can be applied to
recover the shortest paths \cite{GalilMargalit,zwickbridge}. 

\subparagraph*{Total time.}
In both phases, the total cost is 
\begin{eqnarray*}
\OO\left(\max_{\ell}  \MM(n/\ell,n/\ell,n\mid 
\ccc\ell)\right) &= &\OO\left(\max_\ell \ccc\ell \cdot \MM(n/\ell,n/\ell,n)\right)\\
&\le& \OO\left(\max_\ell \ccc\ell^2  (n/\ell)^\omega\right)\ =\ \OO(\ccc n^\omega).
\end{eqnarray*}

\subparagraph*{Remarks.}
Like Shoshan and Zwick's result~\cite{shoshanzwick}, we
can also upper-bound the running time by $\OO(\MM(n,n,n\mid \ccc))$, which is tight,
since it is not difficult to reduce $\MM(n/\ell,n/\ell,n\mid \ccc\ell)$
to $\MM(n,n,n\mid\ccc)$.

The algorithm can be easily modified to solve APSP for a special kind of
directed graphs with weights in $[\ccc]$ that are ``approximately symmetric'',
i.e., $D[u,v]\le s\, D[v,u]$ for every pair $u,v$.  (We change (\ref{tri:ineq}) to
$B[k,j]\le O(s)(A[i,k]+A[i,k']) + B[k',j]$.)
The running time is $\OO(\ccc s n^\omega)$.
This rederives and extends a result by Porat et al.~\cite{PoratST16}, who considered such directed graphs in the unweighted case and obtained an $\OO(sn^\omega)$ time algorithm.

\begin{table}[ht]
    \centering
    \begin{tabular}{|p{0.3\textwidth}|p{0.7\textwidth}|}
    \hline
    Problem & Description\\
    \hline
    Min-Plus product     &  Given an $n\times m$ matrix $A$ and an $m\times p$ matrix $B$, compute the matrix $C$ with entries $C[i,j]=\min_{k=1}^m (A[i,k]+B[k,j])$. \\
    \hline
    $\MMp(n_1,n_2,n_3 \mid M)$ & Compute the Min-Plus product of an $n_1\times n_2$ matrix by an $n_2\times n_3$ matrix where both matrices have integer entries in $[M]$.\\
    \hline
    All-pairs shortest paths (APSP)      &  Compute all-pairs shortest paths distances in a graph. \\
    \hline
    APLP for DAGs & Compute all-pairs longest paths distances in a  directed acyclic graph. \\
    \hline
    $c$Red-APSP & Given a graph in which some edges can be colored red, compute for every pair of vertices $s,t$ the shortest path distance from $s$ to $t$ that uses at most $c$ red edges.\\
    \hline
    Min Witness Equality Product (MinWitnessEq) & Given two $n \times n$ integer matrices $A$ and $B$, compute a matrix $C$ with entries $C[i,j]=\min\{k \in [n]: A[i,k]=B[k, j]\}$. \\
    \hline 
    additive $f(D[u, v])$ approximate APSP & Given a graph where the distance between vertices $u$ and $v$ is $D[u,v]$, compute an estimate $D'[u,v]$ for every $u, v$ such that $D[u,v] \le D'[u,v] \le D[u,v] + f(D[u, v])$.\\
    \hline
    All-Pairs Lightest Shortest Paths (APLSP) & Given a graph, compute for every pair of vertices $s,t$  
the distance from $s$ to $t$ (with respect to the edge weights) and the smallest number of edges over all shortest paths from $s$ and $t$. \\
    \hline
    All-Pairs Shortest Lightest Paths (APSLP) & Given a graph, compute for every pair of vertices $s,t$  
the smallest number of edges of the paths from $s$ to $t$ (with respect to the edge weights) and the smallest length over all such paths from $s$ and $t$.\\
    \hline
    Lex$_2$-APSP & Given a graph where each edge $e$ is given two weights $w_1(e),w_2(e)$, compute for every pair of vertices $u,v$ the lexicographic minimum over all $u$-$v$ paths $\pi$ of $(\sum_{e\in\pi}w_1(e),\sum_{e\in \pi}w_2(e))$. \\
    \hline
    $\#$APSP & Given a graph, count the number of shortest paths for every pair of vertices in a graph.\\
    \hline
    \CountMod{U} & Given a graph, count the number of shortest paths module $U$ for every pair of vertices in a graph.\\
    \hline
     \CountCap{U} & Given a graph, compute the minimum between the number of shortest paths  and $U$ for every pair of vertices in a graph.\\
    \hline
    \CountApx{U} & Given a graph, compute a $(1+1/U)$-approximation of the number of shortest paths for every pair of vertices in a graph.\\
    \hline
    Betweenness Centrality (BC) & Given a graph and a vertex $v$, compute $\text{BC}(v)=\sum_{s,t\neq v} C_v[s,t]/C[s,t]$, where $C[s,t]$ is the number of shortest paths between $s$ and $t$, and $C_v[s,t]$ is the number of shortest paths between $s$ and $t$ that go through $v$.\\
    \hline
    \end{tabular}
    \caption{The problems we consider. For graph problems, we sometimes add a prefix to the problem to constraint the edge weights of the graph. The prefix ``u-'' is for unweighted graphs (e.g. u-APSP);
the prefix ``\ppM-'' is for graphs with weights in \ppM~(e.g. \ppM-APSP), similarly for ``\pmM-'' and for other ranges. All graph problems are defined on directed graphs unless otherwise specified. For instance, u-APSP stands for unweighted directed APSP while u-APSP for undirected graphs stands for unweighted undirected APSP.}
    \label{tab:my_label}
\end{table}

\end{document}